\documentclass[journal,10pt]{IEEEtran}

\usepackage[T1]{fontenc}
\usepackage{amsmath,amssymb,amsthm}
\usepackage{bm}
\usepackage{graphicx}
\usepackage{cite}
\usepackage{microtype}
\usepackage{algorithm}
\usepackage{algorithmic}
\usepackage{pgfplots}
\pgfplotsset{compat=1.16}
\usetikzlibrary{arrows.meta,decorations.pathreplacing}
\pgfplotsset{
  twcaxis/.style={width=\columnwidth,height=4.1cm,
    label style={font=\footnotesize}, tick label style={font=\scriptsize},
    legend style={font=\tiny,fill opacity=0.85,text opacity=1,draw=black!30},
    grid=both, major grid style={black!12}, minor grid style={black!6}}}

\newtheorem{theorem}{Theorem}
\newtheorem{lemma}{Lemma}
\newtheorem{proposition}{Proposition}
\newtheorem{corollary}{Corollary}
\newtheorem{remark}{Remark}

\newcommand{\E}{\mathbb{E}}
\newcommand{\Prb}{\mathbb{P}}
\newcommand{\Var}{\operatorname{Var}}
\newcommand{\Shat}{\widehat S}
\newcommand{\R}{\mathbb{R}}
\newcommand{\odc}{ODC}
\newcommand{\boc}{BOC}

\usetikzlibrary{arrows.meta,backgrounds,positioning,calc}
\definecolor{c1}{RGB}{42,116,196}
\definecolor{c2}{RGB}{237,125,49}
\definecolor{c3}{RGB}{58,158,92}
\definecolor{c4}{RGB}{140,100,200}
\definecolor{pnl}{RGB}{243,248,253}
\definecolor{pbd}{RGB}{190,214,238}
\definecolor{hdc}{RGB}{20,60,120}
\definecolor{bxb}{RGB}{225,238,250}
\definecolor{bxg}{RGB}{228,243,232}
\begin{document}

\title{Occupancy-Domain Over-the-Air Computation}

\author{Seyed~Mohammad~Azimi-Abarghouyi,~\IEEEmembership{Member,~IEEE}%
\thanks{The author is with the Department of Electrical Engineering, Chalmers University of Technology, Gothenburg, Sweden (e-mail: azimimo@chalmers.se).}}

\maketitle

\begin{abstract}
	Over-the-air computation (AirComp) aggregates distributed data through the wireless multiple-access channel, but coherent implementations require channel state information (CSI), phase alignment, and power control, whereas non-coherent energy methods remain affected by fading. Signal superposition at the receive antenna is linear but requires coherence, and energy superposition is linear only in expectation over fading. We introduce occupancy-domain computation (ODC), whose observable is neither a received amplitude nor an energy: the sum is carried by the silence of the shared resources. With exponential Bernoulli activation, the individual silence probabilities multiply, and the server recovers the sum from the idle fraction using binary activity decisions alone, so that once an activation is detected the amplitude that produced it does not enter the estimate. We characterize the maximum-likelihood estimator, optimal load, and a scale-integrated Fisher-information bound for non-adaptive operation over unknown dynamic ranges. We then introduce balanced occupancy computation (BOC), where each device forms a data-dependent quota of random burst placements. This removes the random placement-count fluctuation of Bernoulli activation; under ideal detection, the leading-order asymptotic root-mean-square error of BOC is no larger than that of Bernoulli ODC at any load and approaches $1/\sqrt{2M}$, where $M$ is the number of resource elements, as the number of devices becomes small relative to $M$. We further analyze unknown-scale operation, finite-frame deviations, and heterogeneous detection misses. Simulations validate the theory and compare ODC/BOC with affine non-coherent energy aggregation and REED.
\end{abstract}
\vspace{-5pt}
\vspace{-10pt}
\section{Introduction}
\label{sec:intro}

\IEEEPARstart{M}{any} wireless systems need a function of the values held by their devices rather than the values themselves, and most often that function is a sum. Environmental and industrial sensing, distributed consensus, and federated learning at the network edge all reduce to the same uplink primitive: repeated aggregation over a shared channel. Transmitting the values separately and computing at the server costs one resource per device, so the latency grows with the population. Over-the-air computation (AirComp) avoids this by letting the devices transmit at once, so that the multiple-access channel performs the addition~\cite{NazerGastpar07,GoldenbaumStanczak13,ZhuHuangBAA20,YangShiAirFL20,SahinSurvey23,AzimiSurvey25}. Coherent AirComp achieves this by aligning the effective complex gains at the receiver, so that the received waveform is the desired sum up to noise.

Where that alignment is placed decides what has to be built. Transmitter-side designs burden the devices, which maintain phase and sample-level synchronization and apply power control to keep the aggregate unbiased; those in deep fades cannot meet the required scaling within their budget and are truncated, which biases the result~\cite{CaoZhuPower20,LiuAirCompScaling20}. Receiver-side designs move the burden to the server through denoising-factor optimization, beamforming and device selection~\cite{ZhuHuangBAA20,YangShiAirFL20}, and blind designs go further: the transmitters send without channel knowledge, while the receiver estimates the instantaneous channels and uses a large antenna array to recover the aggregate~\cite{AmiriBlind21,RazavikiaQAM24}. Digital variants change what the constellation carries: computation-by-communications mappings decode a Euclidean structure from an aligned superposition~\cite{Razavikia24}, SumComp gives a closed-form QAM/PAM coding rule over the channel-inverted Gaussian multiple-access channel~\cite{RazavikiaSumComp25}, and one-bit designs decode gradient signs by majority vote~\cite{ZhuOneBit21}. Lattice joint source--channel coding moves part of the burden into a decoding step~\cite{AzimiFedJSCC24,AzimiComputeUpdate24,AzimiAirCPU26}, and learning-aware designs optimize the aggregation weights~\cite{AzimiWeighted24}, the network architecture~\cite{AzimiHierarchical24}, or the propagation environment itself~\cite{AzimiAirPASS26}. All of them are coherent and rely on instantaneous CSI somewhere in the link. This CSI is obtained through pilots and feedback, whose overhead grows with the number of devices and channel selectivity, and must be refreshed once per coherence interval.

These requirements motivate non-coherent AirComp, which gives up exact waveform superposition in exchange for freedom from instantaneous CSI. Type-based multiple access estimates histograms from per-signature energies~\cite{MergenTong06}, and digital balanced-numeral AirComp indexes those signatures by subcarrier position, recovering the aggregate from the energy at each position~\cite{SahinBalanced24}. Continuous sums are obtained from energy statistics or dithering~\cite{GoldenbaumStanczak13,WenNCAirFL24,MichelusiNCOTA}, and a recent unified treatment formalizes this class, non-coherent over-the-air computation (NC-OAC), through affine source-to-energy mappings under statistical or instantaneous channel-amplitude knowledge~\cite{DahlNC26}. REED estimates a continuous signed sum from paired positive and negative energy measurements, with a fading self-noise that repetition over independently faded pairs can average down~\cite{ChenREED26}. These designs remove the need for instantaneous CSI but not for channel knowledge altogether. The affine mappings of~\cite{DahlNC26} and the transmit scaling of~\cite{ChenREED26} both normalize each device by its average channel power, that is, statistical CSI acquired by slow-timescale calibration and tracked as the geometry changes.

Coherent and non-coherent designs alike recover the computation from a real- or complex-valued signal statistic, and this is what limits the non-coherent ones. Energy superposition is linear only in expectation over the fading, so within a coherence block the estimate carries the realized channel gains as random weights, and the resulting error does not vanish with the number of resources spent inside that block. Removing it requires either channel knowledge or diversity across independent fading realizations.

The present work changes the observable itself. Instead of estimating a signal magnitude, the receiver records only whether each resource element is busy or idle. A binary occupancy statistic inherits no channel weight at all: the receiver never reads the amplitude behind a detected burst, so fading acts only through whether detection succeeds. No channel quantity appears in the encoder at any device or in the server-side estimator. The three levels of channel knowledge are therefore distinct: coherent AirComp needs the instantaneous channel, the non-coherent energy designs need only its average power, and occupancy encoding needs neither.

Occupancy statistics are classical in RFID cardinality estimation~\cite{KodialamNandagopal06}, probabilistic counting~\cite{FlajoletMartin85,Whang90}, weighted-cardinality sketches~\cite{Lemiesz21}, and framed random access~\cite{Schoute83,Roberts75,Polyanskiy17}. Collision-based approximate estimation has also been used in large-scale IoT to infer cardinality and related aggregate quantities from random probe collisions~\cite{CaoACE17}, and constant participation budgets are known to improve group-testing designs~\cite{Aldridge19}, where the goal is support recovery. In all of them the randomization is applied by the estimator to data it already holds. Here the randomization is the encoding, so the occupancy statistic is a function of the device values rather than of the population that produced them.

\vspace{-7pt}
\subsection{Contributions}
We develop \emph{occupancy-domain computation} (\odc{}), in which randomized transmissions encode the desired sum in the probability of an idle resource element (RE). Our first construction uses independent Bernoulli activation across the shared REs: each device maps its value into a per-RE transmission probability whose complement depends exponentially on that value, so the silence probabilities multiply across devices and their logarithm recovers the sum. The activation scale thereby becomes a design variable of the computation rather than a property of the channel. We derive the maximum-likelihood estimator, show that the idle count is a sufficient statistic and that the estimator is asymptotically efficient, and obtain the optimal frame load.

We prove a conservation law for non-adaptive designs to show what this encoder cannot escape: the Fisher information a frame carries about the aggregate, integrated over scale, equals $\pi^{2}/6$ however the activation gains are distributed. We use it to bound the worst-case accuracy over a dynamic range, and show that a guarded geometric design operates close to that bound. Since the bound applies to each resource separately, it also indicates how it can be evaded.

Bernoulli activation carries a second cost, independent of that bound: since each RE is decided by its own coin flip, the number of bursts a device transmits is itself random, and that randomness carries no information about the aggregate. We therefore introduce \emph{balanced occupancy computation} (\boc{}), in which each device draws a data-dependent number of burst positions instead of an independent coin per resource, so that its placement-draw quota is fixed up to rounding while the transmissions remain unidentified at the receiver. We show that the occupancy identity remains exact, derive the error law and its operating point, and obtain a constant-factor reduction in relative error, equivalently in channel uses at equal accuracy. We further prove that the leading-order error never exceeds that of Bernoulli activation at any load, and derive a deviation bound valid at finite frame length and finite population.

We propose a two-phase probe for operation without prior scale knowledge, which nearly attains the accuracy available when the scale is known, and a reference-group correction that removes the mean attenuation caused by heterogeneous misses without per-device signal-to-noise-ratio (SNR) knowledge. Fig.~\ref{fig:mechanism} summarizes the common end-to-end mechanism of ODC and BOC. 

We finally evaluate both constructions against non-coherent energy baselines. The margin over affine non-coherent energy aggregation grows with the receive SNR, because that scheme sits at a fading floor that no resource budget can close. The margin over REED is instead bounded, but it is obtained while REED is given exact per-device average channel powers and the occupancy schemes are given none, and a calibration error degrades REED further while leaving the occupancy curves unchanged.

\vspace{-5pt}
\section{System Model and the Occupancy Principle}
\label{sec:model}

\subsection{System model}
Consider $K$ single-antenna devices and a single-antenna server. Device $k$ holds a value $x_k\ge0$, normalized to $[0,1]$ by known bounds as is standard in AirComp. Vector data are handled coordinate-wise, and signed data are treated as described in Section~\ref{sec:ext}. The server estimates the aggregate $S=\sum_{k=1}^{K}x_k$,
or a nomographic function $\Psi(\sum_kg_k(x_k))$ of the device values. 

Resources are organized into frames of $M$ orthogonal REs, for example subcarrier--symbol pairs of an OFDM grid. The channel coefficients $h_{k,m}\in\mathbb C$, for $k=1,\dots,K$ and $m=1,\dots,M$, are unknown to all nodes; no pilots are transmitted and no reciprocity is assumed. Where a comparison with energy-domain schemes requires it, we specialize to block fading, in which $h_{k,m}=h_k$ for all REs of one frame. Occupancy encoding does not use this specialization: Section~\ref{sec:det} works instead with fading independent across REs, which random placement over an interleaved frame makes the relevant case. We assume slot-level alignment only, with frame and RE timing obtained from any beacon-based medium access control protocol and with guard intervals absorbing residual offsets.

Let $\mathcal A_m$ be the set of devices transmitting on RE $m$, and $\psi$ the burst waveform. The received sample is
\begin{equation}
	y_m=\sum_{k\in\mathcal A_m}h_{k,m}\,\psi+z_m,\qquad z_m\sim\mathcal{CN}(0,\sigma_z^{2}).
	\label{eq:rxsignal}
\end{equation}
The receiver forms the energy of $y_m$, normalized to unit noise power, compares it with a threshold $\tau_\alpha$, and returns a binary decision $B_m\in\{0,1\}$, where $B_m=0$ denotes an idle RE; no estimate of $y_m$ itself is formed. The threshold is set from the noise floor so that the false-alarm probability is $\alpha$, and false alarms are independent across REs. Until Section~\ref{sec:det}, detection is ideal and one RE carries one such sample; Section~\ref{sec:det} relaxes both, with $m$ then indexing a group of accumulated samples.

\begin{figure*}[t]
	\centering
	\includegraphics[width=.65\textwidth]{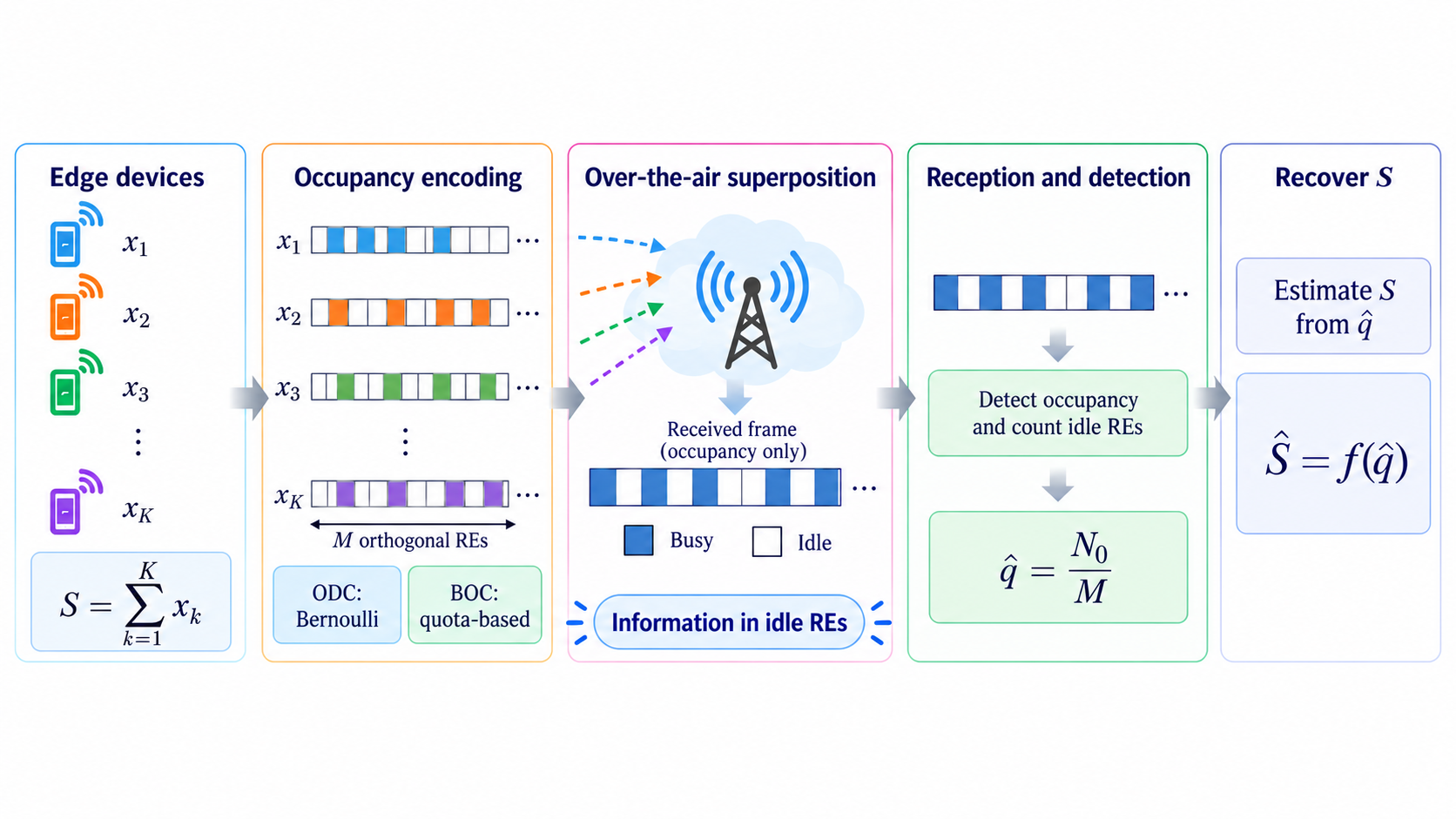}
	\vspace{-40pt}
	\caption{Occupancy-domain over-the-air computation.}
	\vspace{-5pt}
	\label{fig:mechanism}
\end{figure*}

\vspace{-10pt}
\subsection{The occupancy identity}

\begin{lemma}[Exact occupancy superposition]
\label{lem:prod}
Let device $k$ transmit on each RE with probability $p_k$, with the local randomization independent across devices and REs. Then $\Prb\{B_m=0\}=(1-\alpha)\prod_k(1-p_k)$. This is a product over devices, whereas the target $S$ is a sum, so the encoding must make $-\log(1-p_k)$ additive in $x_k$; among monotone maps this forces, up to the choice of a gain $a>0$,
\begin{equation}
	p_k=1-e^{-a x_k},
	\label{eq:encode}
\end{equation}
and then
\begin{equation}
	\Prb\{B_m=0\}=(1-\alpha)\,e^{-aS}.
	\label{eq:identity}
\end{equation}
\end{lemma}
\begin{IEEEproof}
An RE is read as idle if and only if no device activates it and no false alarm occurs. Independence gives the product. For $\Prb\{B_m=0\}$ to depend on the profile only through $S$, the map $g(x)=-\log(1-p(x))$ must satisfy $g(x)+g(y)=G(x+y)$ for all $x,y\ge0$. A device holding $x_k=0$ is silent, so $g(0)=0$ and $G=g$, reducing this to Cauchy's equation, whose monotone solutions are linear: $g(x)=ax$, which is \eqref{eq:encode}. Substituting it turns $\prod_ke^{-ax_k}$ into $e^{-aS}$.
\end{IEEEproof}

The identity of Lemma~\ref{lem:prod} is exact for every finite $K$. It requires no Poisson limit, no Gaussian approximation and no law of large numbers over devices. The channel coefficients do not appear in it, and $K$ never appears separately from $S$, so the number of devices may be unknown and time-varying.

Because the identity involves no channel coefficient, a device in a deep fade and a device close to the server contribute equally once their bursts are detected. The channel gains therefore do not weight the aggregate, and no power control, channel statistics or truncation is required. The transmitted waveform is an arbitrary constant-envelope burst, and on--off keying is sufficient. Amplifier nonlinearity, I/Q imbalance, phase noise and carrier-frequency offset do not enter the mapping \eqref{eq:identity}. They matter only through the busy/idle decisions, which the detection model of Section~\ref{sec:det} accounts for, and through leakage into neighbouring REs, which the guard design discussed there must control. What the transmitter must support is only burst/no-burst operation with a fixed waveform. The cost is spectral efficiency. A coherent scheme uses one RE per aggregated scalar, whereas occupancy encoding spends many REs on a single sum, since the accuracy is set by the number of binary observations; Sections~\ref{sec:scale} and~\ref{sec:boc} quantify this. Against it, occupancy encoding adds no pilot, feedback or calibration budget, whereas the coherent and energy-domain schemes do, so the gap closes as their acquisition overhead grows, fastest when many devices share short coherence intervals.
\vspace{-5pt}
\section{Bernoulli Activation and Its Scale Limitation}
\label{sec:scale}

\subsection{Encoding, decoding and error expression}
Given $M$ and $a$, device $k$ generates for each RE an independent Bernoulli variable with success probability \eqref{eq:encode} and, on success, transmits a constant-envelope burst. Devices with $x_k=0$ remain silent and consume no energy. The server compares each RE with the detection threshold and counts the idle REs $N_0$. Since $N_0/M$ estimates $\Prb\{B_m=0\}$, inverting \eqref{eq:identity} gives
\begin{equation}
\Shat=-\frac{1}{a}\,\log\!\Bigg(\min\!\left\{1,\frac{\max(N_0,1)}{M(1-\alpha)}\right\}\Bigg).
\label{eq:bestimator}
\end{equation}
The transmitter performs $M$ Bernoulli draws, and the receiver performs $M$ threshold comparisons and one logarithm, with no matrix operation and no per-device processing.

Let the \emph{frame load} be $\lambda\triangleq-\ln\Prb\{B_m=0\}$ at $\alpha=0$, so that an RE is idle with probability $e^{-\lambda}$ and $\lambda$ fixes how full the frame is. For this encoder $\lambda=aS$, and with $q=(1-\alpha)e^{-\lambda}$ we have $N_0\sim\mathrm{Bin}(M,q)$.

\begin{theorem}[Sufficiency, efficiency and optimal load]
\label{thm:bern}
$(B_1,\dots,B_M)$ is an i.i.d.\ Bernoulli sample whose distribution depends on the data only through $S$. The count $N_0$ is sufficient, and \eqref{eq:bestimator} is the nonnegative maximum-likelihood estimator (MLE) whenever $N_0\ge1$: the outer minimum returns the boundary estimate $0$ if the empirical idle fraction exceeds the no-activity level $1-\alpha$, while $\max(N_0,1)$ caps the estimate when $N_0=0$. For every fixed $\lambda>0$, both clipping events are asymptotically negligible except that the zero boundary is correct at $S=0$. As $M\to\infty$ at fixed $\lambda>0$,
\begin{equation}
\sqrt M\big(\Shat-S\big)\xrightarrow{d}\mathcal N\Big(0,\frac{1-q}{q}\frac{S^{2}}{\lambda^{2}}\Big),
\label{eq:bclt}
\end{equation}
which attains the Cram\'er--Rao bound of the occupancy model, whose per-RE Fisher information is $a^{2}q/(1-q)$. Consequently, for $\alpha\to0$ the relative root-mean-square error (RMSE) $\rho=\sqrt{\E(\Shat-S)^{2}}/S$ satisfies
\begin{equation}
\rho(\lambda,M)=\frac{1}{\sqrt M}\frac{\sqrt{e^{\lambda}-1}}{\lambda}+o(M^{-1/2}),
\label{eq:brmse}
\end{equation}
which is minimized at the unique positive root $\lambda^\star=1.5936$ of $\lambda e^{\lambda}=2(e^{\lambda}-1)$, where $\rho=(1.2426+o(1))/\sqrt M$. The minimum is shallow: $\lambda=0.8$ and $\lambda=3.2$ increase the error by only $11\%$ and $22\%$.
\end{theorem}
\begin{IEEEproof}
See Appendix~\ref{app:bern}.
\end{IEEEproof}

\begin{remark}[Comparison with the fading-induced error floor]
\label{rem:shot}
The error in \eqref{eq:brmse} is due only to the randomness of the activations, and it decreases without bound as $M$ increases. Energy-superposition schemes under block fading instead contain the term $\sum_k(|h_k|^{2}-\E|h_k|^{2})x_k$, which is fixed within a coherence block and therefore cannot be reduced by repetition inside that block. Under Rayleigh fading normalized to $\E|h_k|^{2}=1$ this gives a relative error floor $\sqrt{\sum_kx_k^{2}}/S=\Theta(1/\sqrt K)$. No such floor arises here, because the observable does not depend on the fading magnitudes.
\end{remark}

At the optimum the idle fraction is $e^{-\lambda^\star}=0.20$ when $\alpha=0$, which is the operating point used in classical occupancy estimation~\cite{KodialamNandagopal06}, here generalized from cardinality estimation to weighted sums. The gain that achieves the optimum, $a=\lambda^\star/S$, depends on the aggregate itself. If the load is too low, almost all REs remain idle and carry little information; if it is too high, idle REs become rare and the estimator becomes unreliable.
\vspace{-5pt}
\subsection{Multi-gain frames}
Choosing the gain in advance therefore means committing to a range of aggregates. The normalization gives $0\le S\le K$, but $S=0$ is possible if every device reports zero, and no fixed gain covers an aggregate arbitrarily close to it. We therefore fix a range $[S_{\min},S_{\max}]$ with $S_{\min}>0$ on which the design is required to perform, and write $D=S_{\max}/S_{\min}$
for its dynamic range. The standard remedy is to use several activation gains within one frame. Divide the frame into $Q$ gain groups with $M_q$ REs each and $\sum_qM_q=M$, and let group $q$ use $a_q=a_0\beta^{q}$ with $\beta>1$. We use equal allocation, $M_q=M/Q$, since no subrange of $[S_{\min},S_{\max}]$ is known to be more likely than another. Two additional groups, referred to as guard groups, are placed beyond the ends of the deployment range so that an aggregate near an endpoint is still covered on both sides. The idle count $N_q$ in group $q$ is $\mathrm{Bin}(M_q,\theta_q(S))$ with $\theta_q(S)=(1-\alpha)e^{-a_qS}$, and the counts $N_q$ are independent and jointly sufficient for $S$.

\begin{proposition}
\label{prop:concave}
The joint log-likelihood
\begin{equation}
\ell(S)=\sum_q\big[N_q\log\theta_q(S)+(M_q-N_q)\log(1-\theta_q(S))\big]
\label{eq:ll}\vspace{-5pt}
\end{equation}
is concave on $S>0$, and strictly concave whenever at least one RE is busy. If every RE is idle, $\ell$ is strictly decreasing and the constrained MLE is $S_{\min}$; if every RE is busy, $\ell$ is strictly increasing and the constrained MLE is $S_{\max}$. Otherwise the constrained MLE on $[S_{\min},S_{\max}]$ is unique. It is the root of $\ell'(S)=\sum_qa_q[(M_q-N_q)\theta_q/(1-\theta_q)-N_q]$ when $\ell'$ changes sign on the interval, obtained by bisection at a cost of $O(Q)$ operations per iteration; it is $S_{\max}$ when $\ell'>0$ throughout and $S_{\min}$ when $\ell'<0$ throughout.
\end{proposition}
\begin{IEEEproof}
Since $\theta_q'(S)=-a_q\theta_q(S)$, we obtain $\ell''(S)=-\sum_q(M_q-N_q)a_q^{2}\theta_q/(1-\theta_q)^{2}\le0$, with strict inequality if any $M_q-N_q>0$. If $N_q=M_q$ for all $q$, then $\ell'(S)=-\sum_qM_qa_q<0$; if $N_q=0$ for all $q$, then $\ell'(S)=\sum_qa_qM_q\theta_q/(1-\theta_q)>0$.
\end{IEEEproof}

All groups must be used jointly. Saturated and nearly empty groups contribute little curvature to \eqref{eq:ll}, but selecting a single group discards $Q-1$ counts, and inverting each group separately is undefined when $N_q=0$.
\vspace{-5pt}
\subsection{A lower bound for non-adaptive designs}
The guarded geometric design is one way to spread the gains over the range, and it is natural to ask whether a better spread exists. Since the accuracy of interest is relative, we work with the information about $\ln S$ rather than about $S$. Multiplying the per-RE Fisher information $a^{2}q/(1-q)$ of Theorem~\ref{thm:bern} by $S^{2}$ and writing $u=aS$, so that $q=e^{-u}$ as $\alpha\to0$, gives the per-RE quantity $J(u)\triangleq u^{2}/(e^{u}-1)$, and the frame information is $\sum_qM_qJ(a_qS)$. A group therefore contributes only where its gain matches the aggregate: $J$ vanishes as $u\to0$ and as $u\to\infty$, and peaks at $u=\lambda^\star$.

\begin{theorem}[Scale-integrated information, lower bound and near-achievability]
\label{thm:conserve}
Consider a non-adaptive Bernoulli design with $\alpha=0$ in which RE $m$ uses gain $a_m$, and let $f$ be the empirical distribution of $\ln a_m$ over the frame. Let $I(s)$ denote the per-RE Fisher information about $\ln S$ evaluated at $S=e^{s}$. Then:
\begin{enumerate}
\item[(i)] for every $f$,
\begin{equation}
\int_{\R}I(s)\,ds=\int_0^\infty J(u)\,\frac{du}{u}=\Gamma(2)\zeta(2)=\frac{\pi^{2}}{6};
\label{eq:conserve}
\end{equation}
\item[(ii)] consequently, over the deployment range $[S_{\min},S_{\max}]$ of dynamic range $D$,
\begin{equation}
\min_{S\in[S_{\min},S_{\max}]}I\;\le\;\frac{\pi^{2}}{6\ln D},
\label{eq:floor}
\end{equation}
so every regular estimator that is unbiased for $\ln S$ has a worst-case relative RMSE over the range of at least $\pi^{-1}\sqrt{6\ln D/M}$;
\item[(iii)] a multi-gain design with geometrically spaced gains and equal allocation achieves, away from finite-range edge effects,
\begin{equation}
\rho(S)=\frac{1}{\pi}\sqrt{\frac{6\,Q\ln\beta}{M}}\big[1+O(\delta_\beta)+o_M(1)\big],
\label{eq:ladderlaw}
\end{equation}
where the log-periodic term satisfies $\delta_2=6.5\times10^{-5}$ and remains at most about $2\times10^{-2}$ up to $\beta=4$. Thus the design approaches the lower bound when $Q\ln\beta=\ln D+o(\ln D)$.
\end{enumerate}
\end{theorem}
\begin{IEEEproof}
See Appendix~\ref{app:conserve}.
\end{IEEEproof}

Expression \eqref{eq:ladderlaw} depends on the design only through $Q\ln\beta$. Gains $a_0\beta^{q}$ for $q=0,\dots,Q-1$ span a ratio $\beta^{Q-1}$, so covering $D$ needs $Q\ge\log_\beta D+1$; adding one guard group at each edge and setting $\beta=2$ gives $Q=\lceil\log_2D\rceil+3$. For $D=10^{3}$ this gives $Q=13$ and the coefficient $2.341/\sqrt M$.

The bound in \eqref{eq:floor} concerns uniform accuracy over a range, not peak accuracy at one operating point. A design that concentrates all its gains at a single scale attains $J(\lambda^\star)=0.648$ nats per RE, which is well above $\pi^{2}/(6\ln D)$, but only in a narrow neighborhood of one value of $S$. The bound also assumes $\alpha=0$: for $\alpha>0$ the constant in \eqref{eq:conserve} becomes the dilogarithm $\mathrm{Li}_2(1-\alpha)<\pi^{2}/6$, so false alarms only reduce the available information and \eqref{eq:floor} is the most favourable case. It further assumes a gain distribution fixed in advance, so it does not apply to a two-phase scheme that first estimates $S$ coarsely; Section~\ref{sec:blind} uses this fact. What \eqref{eq:floor} establishes is the worst-case coefficient $\pi^{-1}\sqrt{6\ln D}$ for non-adaptive independent-gain designs. For $D=10^{3}$ this coefficient is $2.049$, so the guarded geometric design above is within about $14\%$ of the bound; the factor $1.884$ relative to the calibrated optimum is the cost of that particular guarded design, not a universal lower bound.

Since \eqref{eq:floor} is a per-RE Fisher bound, it assumes that the REs are conditionally independent given $S$. The next section removes this assumption.
\vspace{0pt}
\section{Balanced Occupancy Computation}
\label{sec:boc}

\subsection{Quota encoding}
Conditional independence of the REs constrains only how each resource element is decided, and leaves free how a device spreads its transmissions over the frame. We use that freedom: the data determine the number of placement draws, and the local randomness determines only their positions.

Fix a quota gain $a>0$. Device $k$ forms the integer quota
\begin{equation}
n_k=\lfloor ax_k\rfloor+\mathrm{Bern}\big(ax_k-\lfloor ax_k\rfloor\big)
\label{eq:quota}
\end{equation}
and draws $n_k$ RE indices independently and uniformly with replacement from $\{1,\dots,M\}$, activating every distinct selected RE with a constant-envelope burst. Devices with $x_k=0$ remain silent. No exponential mapping and no per-RE random variable are involved, and the transmitter performs $n_k$ index draws instead of $M$ Bernoulli draws. Randomized rounding makes \eqref{eq:quota} conditionally unbiased, with $\E[n_k]=ax_k$ and variance $\sigma_{r,k}^{2}=\phi_k(1-\phi_k)\le1/4$, where $\phi_k$ is the fractional part of $ax_k$. When the fractional parts are spread uniformly over $[0,1)$, $\bar\sigma_r^{2}=K^{-1}\sum_k\sigma_{r,k}^{2}\to\E[\phi(1-\phi)]=1/6$.

If an index is drawn more than once by the same device, one physical burst on that RE is sufficient because the receiver observes only occupancy; hence the number of distinct transmitted bursts is no larger than $n_k$. In Bernoulli activation the data instead determine a per-RE activation probability, and the number of transmitted bursts is itself a random variable, with variance $\sum_kMp_k(1-p_k)$, that carries no information about $S$. This difference is what places the design outside Theorem~\ref{thm:conserve}: with the total $N=\sum_kn_k$ fixed by the quotas, the $M$ occupancy indicators are exchangeable but no longer independent, since they are produced by a shared pool of $N$ draws, and the per-RE information bound of \eqref{eq:floor} does not apply.
\vspace{-10pt}
\subsection{Exact identity and estimator}
\begin{lemma}[Exact quota superposition]
	\label{lem:qidentity}
	Condition on the quotas $\{n_k\}$, so that $N=\sum_kn_k$ is fixed. With bursts placed independently and uniformly and with independent false alarms,
\begin{equation}
\Prb\{B_m=0\mid N\}=(1-\alpha)\Big(1-\frac1M\Big)^{N}
\label{eq:qidentity}
\end{equation}
exactly for every finite $K$ under the ideal-detection abstraction, independently of the received amplitudes and phases.
\end{lemma}
\begin{IEEEproof}
Each placement draw independently misses a given RE with probability $1-1/M$; repeated selections do not change occupancy, and an RE read as idle additionally requires that no false alarm occurs.
\end{IEEEproof}

The occupancy observable depends on the data only through $N$, with $\E[N]=aS$ and $\Var(N)=\sum_k\sigma_{r,k}^{2}\le K/4$. Since $N_0/M$ estimates $\Prb\{B_m=0\mid N\}$, inverting \eqref{eq:qidentity} for $N$ and dividing by $a$ gives the estimator
\begin{equation}
\Shat=\frac1a\cdot\frac{\log\!\Big(\min\!\left\{1,\max(N_0,1)/(M(1-\alpha))\right\}\Big)}{\log(1-1/M)}.
\label{eq:qest}
\end{equation}
Given $N$, the probability of an observed busy/idle pattern is invariant under permutation of the REs, and therefore depends on the pattern only through $N_0$. By the factorization criterion, $N_0$ is sufficient for $N$. Algorithm~\ref{alg:boc} summarizes the operation of one frame.

\begin{algorithm}[t]
\caption{Balanced occupancy computation (one frame)}
\label{alg:boc}
\begin{algorithmic}[1]
	\REQUIRE frame length $M$, quota gain $a$, false-alarm rate $\alpha$
	\STATE each device $k$: form $n_k$ by \eqref{eq:quota}, draw $n_k$ indices i.i.d.\ uniformly on $\{1,\dots,M\}$, and activate each distinct selected RE
	\STATE server: threshold each RE, count the idle REs $N_0$, and output $\Shat$ by \eqref{eq:qest}
\end{algorithmic}
\end{algorithm}
\vspace{-6pt}
\subsection{Error expression and operating point}
Let $\lambda\triangleq\E[N]/M=aS/M$ denote the mean frame load. In \eqref{eq:encode} the gain has units of nats per unit value, whereas in \eqref{eq:quota} it has units of bursts per unit value, so the two gains are not numerically comparable. The idle exponent is, however, common to both: the realized exponent of \eqref{eq:qidentity} at $\alpha=0$ is $-N\log(1-1/M)=(N/M)(1+O(M^{-1}))$, and $N/M=\lambda+O_p(\sqrt K/M)$ by \eqref{eq:quota}, so $\lambda$ matches the quantity used in Section~\ref{sec:scale} up to terms absorbed in the remainder of Theorem~\ref{thm:qlaw}. The encoders are therefore compared at equal $\lambda$.

\begin{theorem}[Error expression for quota encoding]
\label{thm:qlaw}
Assume ideal detection, $K/M=O(1)$, $Me^{-\lambda}\to\infty$ and $M\lambda^{2}\to\infty$. The last two make the truncation in \eqref{eq:qest} inactive with probability tending to one and keep the $O(1)$ remainder of Appendix~\ref{app:qlaw} negligible, while allowing $\lambda\to0$. Then $\Shat/S\to1$ in probability and
\begin{equation}
\rho^{2}=\frac{1}{M\lambda^{2}}\Big[\underbrace{e^{\lambda}-1-\lambda}_{\text{collision noise}}+\underbrace{\frac{K\bar\sigma_r^{2}}{M}}_{\text{quota rounding}}\Big]+o(M^{-1}).
\label{eq:qlaw}
\end{equation}
A nonzero false-alarm rate adds the term $\alpha e^{\lambda}/(1-\alpha)$ inside the bracket. At $\alpha=10^{-3}$ and $\lambda\simeq0.6$ this equals $1.8\times10^{-3}$, against a collision term of $0.22$, and it is therefore included in the simulations but omitted from the design expressions below.
\end{theorem}
\begin{IEEEproof}
See Appendix~\ref{app:qlaw}.
\end{IEEEproof}

\begin{corollary}[Operating point and finite-population expression]
\label{cor:load}
Under the assumptions of Theorem~\ref{thm:qlaw} with $\bar\sigma_r^{2}=1/6$, the optimal load solves
\begin{equation}
e^{\lambda}(\lambda-2)+\lambda+2=\frac{K}{3M},
\label{eq:stat}
\end{equation}
whose positive root satisfies $\lambda^\star=\lambda_0(1+O(\lambda_0))$ with $\lambda_0=(2K/M)^{1/3}$, and
\begin{equation}
\rho\sqrt M=\Big(\tfrac12+\tfrac{\lambda_0}{4}+\tfrac{\lambda_0^{2}}{24}\Big)^{1/2}+O(\lambda_0^{3}),
\label{eq:finiteK}
\end{equation}
so that $\rho\to1/\sqrt{2M}$ as $K/M\to0$.
\end{corollary}
\begin{IEEEproof}
Setting the derivative of $G(\lambda)\triangleq(e^{\lambda}-1-\lambda+c)/\lambda^{2}$ to zero with $c=K\bar\sigma_r^{2}/M$ gives \eqref{eq:stat}. Since $e^{\lambda}(\lambda-2)+\lambda+2=\lambda^{3}/6+\lambda^{4}/12+O(\lambda^{5})$, the leading root of $\lambda^{3}/6=2c=\lambda_0^{3}/6$ is $\lambda_0$, and $\lambda^\star-\lambda_0=O(\lambda_0^{2})$. Evaluating $G$ at $\lambda_0$ with $c=\lambda_0^{3}/12$ gives $G(\lambda_0)=\tfrac12+\tfrac{\lambda_0}{6}+\tfrac{\lambda_0^{2}}{24}+\tfrac{\lambda_0}{12}+O(\lambda_0^{3})$, which is the square of \eqref{eq:finiteK}. Since $\lambda^\star$ minimizes $G$, we have $0\le G(\lambda_0)-G(\lambda^\star)=\tfrac12G''(\xi)(\lambda_0-\lambda^\star)^{2}$ with $G''=O(\lambda_0^{-1})$ in that interval, so the difference is $O(\lambda_0^{3})$.
\end{IEEEproof}

Expression \eqref{eq:finiteK} matches numerical minimization of \eqref{eq:qlaw} to within $1\%$ for $K/M\in[0.02,5]$. Compared with the Bernoulli constant $1.243/\sqrt M$, it predicts a reduction in RMSE of $1.76\times$ as $K/M\to0$, $1.53\times$ at $K/M=0.1$ and $1.32\times$ at $K/M\simeq1$, with corresponding reductions in the number of channel uses of $3.09\times$, $2.34\times$ and $1.75\times$. The asymptotic constants are therefore an upper envelope, and \eqref{eq:finiteK} is the expression that should be compared with simulation.

The residual term $e^{\lambda}-1-\lambda\to\lambda^{2}/2$ is the variance contributed by overlapping bursts, which under independent uniform placement is the dominant remaining error source once the transmitter-side randomness of Bernoulli activation has been removed. The two encoders also differ in transmit activity: at their respective optimal loads, the ratio of the Bernoulli ODC activity budget to the BOC activity budget grows as $(M/K)^{1/3}$. Because $n_k$ is fixed up to rounding, the BOC activity of a device is bounded in advance, which suits duty-cycle-limited radios; a Bernoulli count can fluctuate above its mean. If duty-cycle compliance forces a hard cap on $n_k$, large values are clipped and the aggregate is biased, a regime not analyzed here.

\begin{proposition}[Leading-order comparison with Bernoulli activation at equal load]
\label{prop:dominate}
At equal load $\lambda$ and under the assumptions of Theorem~\ref{thm:qlaw}, the leading $M^{-1}$ terms satisfy
\begin{equation}
\begin{aligned}
\rho_{\boc{}}^{2}(\lambda)-\rho_{\odc{}}^{2}(\lambda)
&=\frac{K\bar\sigma_r^{2}/M-\lambda}{M\lambda^{2}}+o(M^{-1}),\\
\rho_{\odc{}}^{2}(\lambda)
&=\frac{e^{\lambda}-1}{M\lambda^{2}}+o(M^{-1}).
\end{aligned}
\label{eq:dominate}
\end{equation}
Since $K\bar\sigma_r^{2}\le M\lambda$ for every value profile, the leading-order asymptotic RMSE of BOC is no larger than that of Bernoulli ODC at any $\lambda>0$, with a strict reduction in the leading coefficient whenever $S>0$.
\end{proposition}
\begin{IEEEproof}
Subtract $(e^{\lambda}-1)/(M\lambda^{2})$ from \eqref{eq:qlaw} using $e^{\lambda}-1=(e^{\lambda}-1-\lambda)+\lambda$, which gives the leading term in \eqref{eq:dominate}. For the sign of that term, $\sigma_{r,k}^{2}=\phi_k(1-\phi_k)\le ax_k$ for every $k$: if $ax_k\ge1$ then $\phi_k(1-\phi_k)\le1/4<1\le ax_k$, and if $ax_k<1$ then $\phi_k=ax_k$ and $\phi_k(1-\phi_k)\le\phi_k=ax_k$. Summing, $K\bar\sigma_r^{2}\le aS=M\lambda$, with equality only if $ax_k=0$ for all $k$.
\end{IEEEproof}

\begin{remark}[Regime of validity]
\label{rem:regime}
When $\max_kax_k\le1$, the quota in \eqref{eq:quota} reduces to $\mathrm{Bern}(ax_k)$, which is a linear rather than an exponential activation mapping, and $K\bar\sigma_r^{2}=\E[N]-\sum_k(ax_k)^{2}$, which is the case of Proposition~\ref{prop:dominate} in which the gap $\sum_k(ax_k)^{2}$ is smallest. \boc{} therefore retains a strictly smaller leading-order error coefficient than Bernoulli activation, but the difference decreases as $\max_kax_k\to0$, which for value profiles without a dominant device means large $K/M$. The advantage is largest around $K\lesssim M$ and decreases smoothly beyond it, within the range $K/M=O(1)$ in which \eqref{eq:qlaw} was derived.
\end{remark}
\vspace{-10pt}
\subsection{A non-asymptotic guarantee}
Theorem~\ref{thm:qlaw} is asymptotic in $M$, whereas the identity \eqref{eq:qidentity} on which it is based is exact for every finite $K$ and $M$. The following bound holds for finite values of both, and it also accounts explicitly for the event $N_0=0$.

\begin{theorem}[Finite-frame deviation bound]
\label{thm:finite}
Assume ideal detection and $\alpha=0$, and let $q=(1-1/M)^{N}$. For every $t>0$, $\Prb\{|N_0-Mq|\ge t\mid N\}\le2e^{-2t^{2}/N}$. Consequently, for any $\delta\in(0,1)$, let $t_\delta=\sqrt{(N/2)\ln(4/\delta)}$ and $s_\delta=\sqrt{(K/2)\ln(4/\delta)}$. If $t_\delta<Mq$, then with probability at least $1-\delta$ we have $N_0>0$ and
\begin{equation}
\frac{|\Shat-S|}{S}\;\le\;\frac{M+1}{M\lambda}\cdot\frac{t_\delta/(Mq)}{1-t_\delta/(Mq)}\;+\;\frac{s_\delta}{M\lambda},
\label{eq:finitebound}
\end{equation}
which for fixed $\lambda$ and $M\to\infty$ behaves as $\frac{e^{\lambda}}{\lambda}\sqrt{\lambda\ln(4/\delta)/(2M)}+\frac{1}{\lambda}\sqrt{K\ln(4/\delta)/(2M^{2})}$. Since $t_\delta$ and $q$ are functions of $N$, the right-hand side of \eqref{eq:finitebound} is random; the union bound applies conditionally on $N$ to the first event and unconditionally to the second.
\end{theorem}
\begin{IEEEproof}
See Appendix~\ref{app:finite}.
\end{IEEEproof}

The two terms in \eqref{eq:finitebound} correspond to the randomness of the burst placement and to the quota rounding, matching the two terms in \eqref{eq:qlaw}. The bound is looser than the asymptotic expression, giving about $4.3/\sqrt M$ at $K/M=0.1$ and $\delta=0.05$ against $1.6/\sqrt M$ for a Gaussian $95\%$ interval. The loss is the usual cost of a bounded-difference argument, and the bound holds for every $M$ and $K$.
\vspace{-5pt}
\section{Selecting the Gain Without Scale Knowledge}
\label{sec:blind}

Algorithm~\ref{alg:boc} takes the gain $a$ as an input, and by Corollary~\ref{cor:load} the appropriate value depends on $S$. Theorem~\ref{thm:conserve} applies only to non-adaptive designs, so the scale can be acquired first and the gain set from the result. We split the frame into two phases: $M_1$ REs that probe for $S$, and the remaining $M-M_1$ that carry the computation. Phase~1 uses a guarded multi-gain Bernoulli design with the estimator of Proposition~\ref{prop:concave}. The probe need only locate $S$ to within tens of percent, for the reason quantified by \eqref{eq:mismatch} below, so $M_1$ can be a few percent of $M$ and the number of groups far smaller than the $Q$ of \eqref{eq:ladderlaw}. If the gain is set for an aggregate $S/r$ rather than $S$, the actual load becomes $r\lambda^\star$. Under the design approximation $\bar\sigma_r^2\simeq1/6$, so that $c=K/(6M)$ is locally fixed, the squared-error ratio is
\begin{equation}
\frac{\rho^{2}(r\lambda^\star)}{\rho^{2}(\lambda^\star)}=\frac{\big(e^{r\lambda^\star}-1-r\lambda^\star+c\big)/(r\lambda^\star)^{2}}{\big(e^{\lambda^\star}-1-\lambda^\star+c\big)/\lambda^{\star2}},\quad c=\frac{K}{6M}.
\label{eq:mismatch}
\end{equation}
Evaluating \eqref{eq:mismatch} at $K/M=0.1$ gives an RMSE increase of $6.6\%$ for $r=2$ and $8.7\%$ for $r=1/2$. 

The server then broadcasts the scalar $a$ in the beacon it already transmits for slot timing, and Phase~2 runs Algorithm~\ref{alg:boc} on the remaining $M-M_1$ REs. No per-device feedback, instantaneous CSI or device identification is involved. The factor $1.884$ that the guarded design pays on every frame is thereby replaced by the probe overhead and one downlink broadcast. Without such a beacon, the broadcast is an added requirement.

When several short downlink turnarounds are available, Phase~1 can instead perform a bisection search on $\ln S$. The server probes with $a=\lambda^\star/\sqrt{LU}$ over a multiplicative interval $[L,U]$, initialized to $[S_{\min},S_{\max}]$, and compares the idle fraction with $e^{-\lambda^\star}=0.2032$. Since $e^{-aS}$ is monotone in $S$, the sign of this difference identifies the correct half of the interval, and each such decision replaces $U/L$ by its square root. A residual mismatch $r$ is therefore reached in $O(\log\log D)$ stages, which is four stages for $D=10^{4}$ and $r=2$. All probe counts still enter \eqref{eq:ll}, so no probe RE is discarded. This variant replaces one broadcast by several and is preferable only when the downlink latency is small.
\vspace{-5pt}
\section{Imperfect Detection}
\label{sec:det}
\vspace{-3pt}
\subsection{The erasure model and its validity}
Theorems~\ref{thm:qlaw} and~\ref{thm:finite} assume that every burst is detected. A burst in a deep fade may fail the threshold test, and since the estimator counts idle REs, an undetected burst is indistinguishable from a device that never transmitted. We therefore model detection at the level of \emph{effective activations}. An activation by device $k$ is effective with probability $1-\varepsilon_k$, independently across REs, across bursts of the same device, and across devices, where $\varepsilon_k$ is the per-burst miss probability determined by the average receive SNR $\bar\gamma_k=\E|h_{k,m}|^{2}|\psi|^{2}/\sigma_z^{2}$ of device $k$, the threshold $\tau_\alpha$ and the fading distribution. Independence across a device's own bursts assumes the frame is interleaved over enough coherence dimensions that scattered REs fade independently; uniform placement facilitates this but does not by itself guarantee it. If the bursts instead shared one coherence block, $\varepsilon_k$ would act on the device as a whole rather than per burst. An RE is read as busy if it carries an effective activation or a false alarm. For later use, define the data-weighted miss moments
\begin{equation}
\varepsilon_x\triangleq\frac{\sum_kx_k\varepsilon_k}{S},\qquad
\varepsilon_{2,x}\triangleq\frac{\sum_kx_k\varepsilon_k^2}{S},
\label{eq:missmom}
\end{equation}
and the unweighted moments $\varepsilon_u=K^{-1}\sum_k\varepsilon_k$ and $\varepsilon_{2,u}=K^{-1}\sum_k\varepsilon_k^2$.

\begin{lemma}[Effect of independent misses]
\label{lem:miss}
Under Bernoulli activation,
\begin{equation}
\Prb\{B_m=0\}=(1-\alpha)\prod_k\big[\varepsilon_k+(1-\varepsilon_k)e^{-ax_k}\big],
\label{eq:bernmiss}
\end{equation}
so the estimator \eqref{eq:bestimator}, applied without correction, converges to
\begin{equation}
\begin{aligned}
S_{B,\varepsilon}(a)
&=-\frac1a\sum_k\log\!\big[\varepsilon_k+(1-\varepsilon_k)e^{-ax_k}\big]\\
&=\sum_k(1-\varepsilon_k)x_k+O\!\Big(a\sum_kx_k^2\Big).
\end{aligned}
\label{eq:bernmisslimit}
\end{equation}
when $\max_k ax_k\to0$. Under quota encoding with $\max_kn_k=o(M)$, independent thinning gives the asymptotic target $S_\varepsilon=\sum_k(1-\varepsilon_k)x_k$. Let $\lambda_e=aS_\varepsilon/M=\lambda(1-\varepsilon_x)$ and $\rho_e^2\triangleq\E[(\Shat-S_\varepsilon)^2]/S_\varepsilon^2$. Under the scaling of Theorem~\ref{thm:qlaw},
\begin{equation}
\rho_e^{2}=\frac{1}{M\lambda_e^{2}}\Big[e^{\lambda_e}-1-\lambda_e+\frac{\Sigma_r}{M}+\lambda(\varepsilon_x-\varepsilon_{2,x})\Big]+o(M^{-1}),
\label{eq:thin}
\end{equation}
where $\Sigma_r=\sum_k(1-\varepsilon_k)^{2}\sigma_{r,k}^{2}$.
\end{lemma}
\begin{IEEEproof}
See Appendix~\ref{app:miss}.
\end{IEEEproof}

Weighted Jensen gives $\varepsilon_{2,x}\ge\varepsilon_x^{2}$, so $\varepsilon_x-\varepsilon_{2,x}\le\varepsilon_x(1-\varepsilon_x)$, with equality for homogeneous miss probabilities, and \eqref{eq:qlaw} is recovered as the miss probabilities vanish. Misses bias both uncorrected encoders; for BOC they additionally introduce the explicit binomial-thinning variance term in \eqref{eq:thin}.

Independence of the erasures across devices is the one part of this model that the detector does not automatically deliver, and both encoders rely on it. Suppose the coefficients $h_{k,m}$ are Rayleigh, so that $y_m$ is complex Gaussian, and let a logical burst span $N_s$ independently faded dimensions. With $\mathcal A_m$ the transmitting set of \eqref{eq:rxsignal}, the detector energy $Z$ then satisfies $Z\mid\mathcal A_m\sim\mathrm{Gamma}(N_s,1+\sum_{k\in\mathcal A_m}\bar\gamma_k)$, and the true idle probability is $\E_{\mathcal A_m}\big[\Prb\{Z<\tau_\alpha\mid\mathcal A_m\}\big]$, which need \emph{not} equal \eqref{eq:identity}. For $N_s=1$, all co-activating devices are evaluated through a single random energy sample, so their miss events are correlated. Increasing $N_s$ concentrates the energy statistic and reduces the mismatch with the independent-erasure abstraction, at an explicit cost. Call the $N_s$ physical channel uses carrying one busy/idle decision a \emph{logical RE}: $M$ of them consume $MN_s$ channel uses, so a fixed budget of $R$ channel uses gives only $M=R/N_s$ logical REs. Since $\rho$ scales as $M^{-1/2}$ by \eqref{eq:qlaw}, the error acquires a factor $\sqrt{N_s}$. Section~\ref{sec:num} reports the resulting RMSE as a function of $N_s$.

Timing and interference impairments enter only through the busy/idle decisions. A burst delayed within the guard interval still marks its own RE as busy; a larger delay marks a neighbouring RE instead, which biases the count and must be controlled by the guard design.

An RE subject to persistent interference appears always busy and can be excluded before counting, provided the exclusion rule is independent of the device data, since \eqref{eq:qlaw} then applies with $M$ replaced by the number of retained REs.
\vspace{-10pt}
\subsection{Pre-compensation for Bernoulli activation}
A device using Bernoulli activation can compensate an estimated long-term miss probability $\hat\varepsilon_k$ by using
\begin{equation}
p_k^{\rm pc}=\frac{1-e^{-ax_k}}{1-\hat\varepsilon_k},
\label{eq:precomp}
\end{equation}
provided $p_k^{\rm pc}\le1$. If $\hat\varepsilon_k=\varepsilon_k$, then under the independent-erasure model $1-(1-\varepsilon_k)p_k^{\rm pc}=e^{-ax_k}$ and the ideal occupancy identity is restored exactly. The required $\hat\varepsilon_k$ can be estimated from long-term link statistics, so no instantaneous CSI is needed. If $p_k^{\rm pc}$ exceeds one, it must be clipped, leaving residual bias. Moreover, when several devices activate the same RE, their energies add non-coherently, so joint detection can be more likely than the independent-erasure model predicts; pre-compensation can then over-correct. It is therefore most appropriate for moderate miss probabilities and operating points for which $p_k^{\rm pc}$ remains well below one. This correction is analyzed but not evaluated numerically; the detection study of Section~\ref{sec:num} uses quota encoding with the self-normalization of Section~\ref{ssec:selfnorm} instead.
\vspace{-5pt}
\subsection{Self-normalization for quota encoding}
\label{ssec:selfnorm}
Pre-compensation requires $\bar\gamma_k$ and $\tau_\alpha$. The self-normalized quota construction below requires neither, but it estimates the mean, so recovering the sum requires $K$ to be known. Divide the frame into a data group of $M_d$ REs carrying quotas with $\E[n_k]=ax_k$, and a reference group of $M_r$ REs in which \emph{every} device transmits the same constant quota $n_0$. Let $N_{0,d}$ and $N_{0,r}$ be the idle counts and define the occupancy-inverted placement counts
$\widehat N_\bullet=\log(\min\{1,\max(N_{0,\bullet},1)/(M_\bullet(1-\alpha))\})/\log(1-1/M_\bullet)$ for $\bullet\in\{d,r\}$.

\begin{proposition}[Mean estimation without SNR knowledge]
\label{prop:selfnorm}
Let $M_d,M_r\to\infty$ with positive effective group loads held fixed and with $K=O(M_d)$, and let $n_0K\to\infty$. Then
\begin{equation}
\widehat{\bar x}\triangleq\frac{\widehat N_d/a}{\widehat N_r/n_0}\longrightarrow\frac{\sum_k(1-\varepsilon_k)x_k}{\sum_k(1-\varepsilon_k)}=\bar x+\frac{\widehat{\operatorname{Cov}}(1-\varepsilon,x)}{1-\varepsilon_u},
\label{eq:ratio}
\end{equation}
where $\bar x=K^{-1}\sum_kx_k$ and $\widehat{\operatorname{Cov}}(1-\varepsilon,x)=K^{-1}\sum_k(1-\varepsilon_k)x_k-(1-\varepsilon_u)\bar x$ is the empirical covariance over the $K$ devices. The limit holds for arbitrary, unknown and heterogeneous $\{\varepsilon_k\}$ and requires no SNR knowledge. If in addition $(\varepsilon_k,x_k)$ are i.i.d.\ with $\varepsilon_k$ independent of $x_k$, then $\E[\widehat{\operatorname{Cov}}]=0$ and the residual term is $O_p(K^{-1/2})$, so $\widehat{\bar x}$ is consistent for $\E[x]$ as $K\to\infty$.
\end{proposition}
\begin{IEEEproof}
Thinning a uniformly placed quota by independent Bernoulli erasures yields uniformly placed surviving bursts, so Lemma~\ref{lem:qidentity} applies in both groups conditional on the thinned count $\widetilde N$. At fixed positive effective load, Theorem~\ref{thm:qlaw} gives $\widehat N_\bullet/\widetilde N_\bullet\to1$ in probability. For the data group, Appendix~\ref{app:miss} gives $\E[\widetilde N_d]=a\sum_k(1-\varepsilon_k)x_k=\Theta(M_d)$ and $\Var(\widetilde N_d)=O(M_d+K)=O(M_d)$, hence $\widetilde N_d/\E[\widetilde N_d]\to1$. For the reference group, $\E[\widetilde N_r]=n_0\sum_k(1-\varepsilon_k)$ and $\Var(\widetilde N_r)\le n_0K/4$; with a nonvanishing average detection probability and $n_0K\to\infty$, its relative fluctuation also vanishes. Therefore the numerator and denominator of \eqref{eq:ratio}, after division by $K$, converge to $K^{-1}\sum_k(1-\varepsilon_k)x_k$ and $K^{-1}\sum_k(1-\varepsilon_k)$, respectively. Adding and subtracting $(1-\varepsilon_u)\bar x$ gives the second equality. The last claim follows from the central limit theorem applied to $\widehat{\operatorname{Cov}}$.
\end{IEEEproof}

Independence between $\varepsilon_k$ and $x_k$ is reasonable when detectability is governed mainly by geometry while the value is governed by the sensed quantity.

\begin{proposition}[Resource overhead of self-normalization]
\label{prop:split}
Let both groups operate at their own optimal loads with $K/M\to0$, so that each attains $\rho=1/\sqrt{2M_\bullet}$, and let the groups be independent. Since a ratio of independent estimates has the sum of their relative variances, the delta method gives $\rho_{\rm ratio}^{2}\simeq\rho_d^{2}+\rho_r^{2}=\tfrac12(M_d^{-1}+M_r^{-1})$, which is minimized over $M_d+M_r=M$ at $M_d=M_r=M/2$, where $\rho_{\rm ratio}=2/\sqrt{2M}$. Self-normalization therefore increases the ideal-detection RMSE by a factor of $2$ relative to using all $M$ REs for a single BOC estimate. The unequal choice $M_r=M/4$ gives a factor $\sqrt{16/3}=2.31$.
\end{proposition}

Under ideal detection, the reference group should use the smallest nonzero quota, $n_0=1$. To see this, define $g(\lambda)\triangleq(e^{\lambda}-1-\lambda)/\lambda^{2}$, so that the reference-group error is $\sqrt{g(\lambda_r)/M_r}$ at load $\lambda_r=n_0K/M_r$. Then $g'(\lambda)=h(\lambda)/\lambda^{3}$ with $h(\lambda)=\lambda(e^{\lambda}-1)-2(e^{\lambda}-1-\lambda)$, and $h(0)=h'(0)=0$ with $h''(\lambda)=\lambda e^{\lambda}>0$, so $g$ is strictly increasing and the error grows with $n_0$.

This conclusion changes under nonzero misses, which is the regime in which self-normalization is used. For the reference group, substituting $N_r=n_0K$ into \eqref{eq:thin} and letting $\lambda_r\to0$ gives
\begin{equation}
\rho_r^{2}=\frac{1}{2M_r}+\frac{\varepsilon_u-\varepsilon_{2,u}}{n_0K(1-\varepsilon_u)^{2}}+O(\lambda_r).
\label{eq:refquota}
\end{equation}
In \eqref{eq:refquota} the collision term increases with $n_0$ through $\lambda_r=n_0K/M_r$ while the thinning term decreases as $1/n_0$, so an interior optimum exists. At $M=2048$, $K=200$, $M_r=M/4$ and $\varepsilon_k\sim\mathrm{Unif}[0,2\bar\varepsilon]$, evaluating \eqref{eq:thin} gives $n_0=1$ for $\bar\varepsilon\le0.05$, $n_0=2$ at $\bar\varepsilon=0.15$ and $n_0=4$ at $\bar\varepsilon=0.30$. The symmetric split of Proposition~\ref{prop:split} remains the ideal-detection benchmark; with nonzero misses, the optimal split can shift because the thinning terms depend on the two group sizes.
\vspace{-5pt}
\section{Extensions}
\label{sec:ext}

\subsection{Function class and signed data}
Replacing $x_k$ by $g_k(x_k)\ge0$, nonnegativity being required by \eqref{eq:encode} and \eqref{eq:quota}, and applying $\Psi$ at the server computes $\Psi(\sum_kg_k(x_k))$ with the accuracy of Theorem~\ref{thm:bern} or Theorem~\ref{thm:qlaw} applied to $S_g=\sum_kg_k(x_k)$. This is the nomographic function class computable by analog AirComp~\cite{Buck79,GoldenbaumStanczak13}. Since $B_m$ is the logical OR of the activations on RE $m$, an additional burst can only turn an idle RE busy, so no device can cancel another's contribution. Handling signed data therefore requires two RE groups $A$ and $B$ carrying $x_k^{+}=\max(x_k,0)$ and $x_k^{-}=\max(-x_k,0)$, with $\Shat=\Shat_A-\Shat_B$. This doubles the RE cost, and the variances of the two estimates add. Vector aggregation remains linear in the dimension; superimposing several coordinates on one frame through random nonnegative weights preserves the exact identity and yields a concave likelihood in the aggregate vector, but is not developed here.
\vspace{-5pt}
\subsection{Privacy}
\label{sec:privacy}
Occupancy encoding leaks less per observation than amplitude or energy schemes, where one isolated RE already carries $x_k$. Between the two encoders, Bernoulli activation randomizes each device's transmission pattern, whereas quota encoding fixes the placement count up to rounding. An adversary able to isolate the transmissions of one device therefore recovers $ax_k$ to within one burst under quota encoding, but observes only a $\mathrm{Bin}(M,p_k)$ sample under Bernoulli activation. Determinism of the placement count is thus the one respect in which quota encoding is the weaker of the two. In the intended over-the-air setting the receiver observes only the superposition of anonymous transmissions, which limits per-device inference but does not constitute a formal guarantee. Differential privacy can be obtained by adding calibrated randomization to the activation or quota mechanism~\cite{DworkRoth14}, at the cost of additional aggregation error. The resulting privacy--utility trade-off is not analyzed here.
\vspace{-5pt}
\subsection{Federated learning}
Federated SGD repeats the same aggregation every round, so the per-round savings of Section~\ref{sec:model} compound. Gradients are signed and high-dimensional, so an extension would clip them coordinate-wise, aggregate by the signed encoding above, and apply a shared count-sketch projection~\cite{Charikar02} first, since linear sketches commute with summation. The aggregation error would enter the SGD bound of~\cite{GhadimiLan13} as a variance term falling with $M$ by \eqref{eq:brmse} or \eqref{eq:finiteK}, and residual detection bias as a neighborhood term~\cite{AjalloeianStich20}, so the frame length could be traded against accuracy per round. This is left for future work.
\vspace{-5pt}
\section{Numerical Results}
\label{sec:num}

Unless stated otherwise, the device values are drawn independently as $x_k\sim\mathrm{Unif}[0,1]$, fresh activation patterns and quota placements are generated in every frame, and performance is reported as the relative RMSE $\rho$ of Theorem~\ref{thm:bern}. Ideal-detection experiments use $\alpha=0$ and need no channel realization, since \eqref{eq:identity} contains none. Detector experiments use $\alpha=10^{-3}$ and Rayleigh fading that is independent across REs and devices, with the average receive powers $\bar\gamma_k$ drawn log-uniformly over a $10$~dB spread about the stated mean receive SNR. In Figs.~\ref{fig:qlaw} and~\ref{fig:scaling} markers show simulation and lines show theory; elsewhere curves are simulated and reference lines are labelled. 

For the external comparison, we restrict the numerical baselines to non-coherent continuous-sum methods that operate without instantaneous CSI: the unit-affine energy estimator in the NC-OAC framework of~\cite{DahlNC26} and REED~\cite{ChenREED26}. Both use statistical channel powers at the transmitter, which ODC and BOC do not. The affine NC-OAC reference uses unit statistical channel-power normalization over a block-Rayleigh channel, and REED retains its native paired measurements, with both branches counted in the total physical-RE budget. Coherent and digital AirComp schemes are not included numerically because they require instantaneous CSI somewhere in the link, whereas ODC and BOC target operation without channel knowledge of any kind. Sign-aggregation and learning-specific schemes discussed in Section~\ref{sec:intro} are likewise excluded because they target different outputs or optimize an end-to-end learning objective rather than the same continuous scalar sum.
\vspace{-5pt}
\subsection{Theory, scale uncertainty, and population scaling}

\begin{figure}[t]
\centering
\begin{tikzpicture}
\begin{loglogaxis}[twcaxis, xlabel={aggregate $S$}, ylabel={relative RMSE},
 xmin=1.5,xmax=7000, ymin=0.02, ymax=1.05,
 legend style={font=\scriptsize,fill opacity=0.92,text opacity=1,at={(0.02,0.98)},anchor=north west}]
\addplot[mark=triangle*,brown,very thick,mark size=1.45] coordinates {(1.7783,0.05987)(4.8697,0.06052)(13.3352,0.05887)(36.5174,0.06020)(100,0.05956)(273.8420,0.06008)(749.8942,0.06065)(2053.5250,0.06138)(5623.4133,0.06405)};
\addlegendentry{multi-gain design}
\addplot[mark=square*,black,very thick,mark size=1.35] coordinates {(1.7783,0.02761)(4.8697,0.02741)(13.3352,0.02775)(36.5174,0.02801)(100,0.02726)(273.8420,0.02745)(749.8942,0.02788)(2053.5250,0.02752)(5623.4133,0.02766)};
\addlegendentry{known scale}
\addplot[mark=diamond*,red,densely dashed,very thick,mark size=1.35] coordinates {(1.7783,0.13368)(4.8697,0.08044)(13.3352,0.04999)(36.5174,0.03388)(100,0.02726)(273.8420,0.04560)(749.8942,0.36198)(2053.5250,0.76701)(5623.4133,0.91492)};
\addlegendentry{fixed gain, $S_{\rm ref}=100$}
\end{loglogaxis}
\end{tikzpicture}
\vspace{-3pt}
\caption{Scale uncertainty for Bernoulli activation, $M=2048$, $D=10^4$ ($4000$ frames per point).}
\label{fig:scale}
\end{figure}
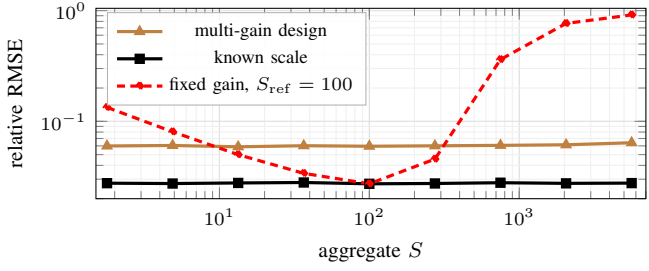

Fig.~\ref{fig:scale} uses Bernoulli activation throughout and compares the guarded multi-gain design with a known-scale reference and with a single fixed gain calibrated at $S_{\rm ref}=100$. The aggregate is swept over four decades rather than generated from the default value distribution, and the scale limitation of Theorem~\ref{thm:conserve} is visible. Across that range the multi-gain design holds the RMSE between $5.9\%$ and $6.4\%$, against $2.75\%$ for the known-scale benchmark. The near-constant level is the signature of \eqref{eq:ladderlaw}: spreading the gains geometrically flattens the per-RE information $I(s)$ across the range, and because the scale integral of $I$ is fixed by \eqref{eq:conserve}, flattening it necessarily lowers its peak. Uniformity over the range is thus paid for by a fixed multiplicative loss in accuracy. The fixed-gain estimator instead spends all of its information near $S_{\rm ref}$ and collapses on either side, which is why a single open-loop gain is usable only when the aggregate is already known to within the tolerance of \eqref{eq:mismatch}.

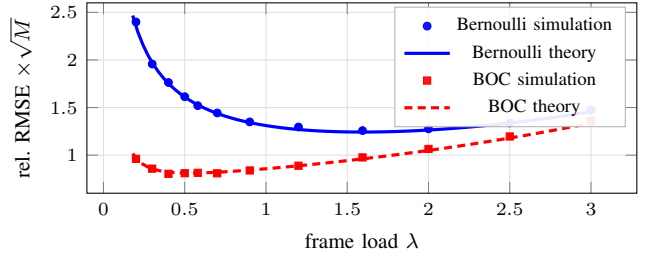
\begin{figure}[t]
\centering
\begin{tikzpicture}
\begin{axis}[twcaxis, xlabel={frame load $\lambda$}, ylabel={rel. RMSE $\times\sqrt M$},
 ymin=0.6, ymax=2.6,
 legend style={font=\scriptsize,fill opacity=0.92,text opacity=1,at={(0.98,0.98)},anchor=north east}]
\addplot[blue,mark=*,mark size=1.45,only marks] coordinates
{(0.2,2.399)(0.3,1.957)(0.4,1.763)(0.5,1.614)(0.58,1.520)(0.7,1.443)(0.9,1.349)(1.2,1.296)(1.594,1.258)(2,1.276)(2.5,1.338)(3,1.474)};
\addlegendentry{Bernoulli simulation}
\addplot[blue,very thick,domain=0.18:3,samples=90]{sqrt(exp(x)-1)/x};
\addlegendentry{Bernoulli theory}
\addplot[red,mark=square*,mark size=1.35,only marks] coordinates
{(0.2,0.961)(0.3,0.858)(0.4,0.804)(0.5,0.810)(0.58,0.814)(0.7,0.809)(0.9,0.839)(1.2,0.889)(1.594,0.976)(2,1.065)(2.5,1.196)(3,1.357)};
\addlegendentry{BOC simulation}
\addplot[red,very thick,densely dashed,domain=0.18:3,samples=90]{sqrt(exp(x)-1-x+0.016276)/x};
\addlegendentry{BOC theory}
\end{axis}
\end{tikzpicture}
\vspace{-3pt}
\caption{Relative RMSE versus frame load, $K=100$, $M=1024$ ($6000$ trials).}
\label{fig:qlaw}
\end{figure}

Fig.~\ref{fig:qlaw} compares the two error laws with simulation across the full load range. The simulated points follow both laws, including near their minima. Bernoulli activation reaches $\rho\sqrt M\simeq1.24$ near $\lambda=1.59$; BOC reaches about $0.81$ near $\lambda=0.53$. The two minima differ in position as well as depth, and both follow from the composition of the error. Bernoulli activation pays $e^{\lambda}-1$, whose activation-count term grows only linearly, so the optimum sits where the $1/\lambda^{2}$ normalization stops repaying it. Quota encoding removes that term, leaving the collision and rounding terms of \eqref{eq:qlaw}, and the residual $e^{\lambda}-1-\lambda\to\lambda^{2}/2$ is quadratic rather than linear near the origin. The optimum therefore moves to a sparser frame, here to a third of the Bernoulli load, and by Corollary~\ref{cor:load} the ratio $\lambda_0/\lambda^\star$ falls further as $K/M$ decreases. Running BOC at the Bernoulli optimum would forfeit much of the gain: the curve there reads $0.98$ against $0.81$ at its own optimum.

\begin{figure}[t]
\centering
\begin{tikzpicture}
\begin{semilogxaxis}[twcaxis, xlabel={population load $K/M$}, ylabel={rel. RMSE $\times\sqrt M$},
 xmin=0.015,xmax=2.3,ymin=0.68,ymax=1.38,
 legend style={font=\scriptsize,fill opacity=0.92,text opacity=1,at={(0.5,1.02)},anchor=south,legend columns=2,
 /tikz/every even column/.append style={column sep=0.25cm}}]
\addplot[blue,mark=*,only marks,mark size=1.45] coordinates {(0.0195312,1.28232)(0.0488281,1.23018)(0.0976562,1.24918)(0.195312,1.21958)(0.488281,1.27996)(0.976562,1.25926)(1.95312,1.26881)};
\addlegendentry{Bernoulli sim.}
\addplot[blue,very thick] coordinates {(0.0195,1.24263)(2.0,1.24263)};
\addlegendentry{Bernoulli theory}
\addplot[red,mark=square*,only marks,mark size=1.35] coordinates {(0.0195312,0.77251)(0.0488281,0.78307)(0.0976562,0.80790)(0.195312,0.84749)(0.488281,0.86107)(0.976562,0.92573)(1.95312,1.01507)};
\addlegendentry{BOC sim.}
\addplot[red,very thick,densely dashed] coordinates {(0.0195312,0.76790)(0.0488281,0.79000)(0.0976562,0.81200)(0.195312,0.83996)(0.488281,0.88901)(0.976562,0.93817)(1.95312,1.00105)};
\addlegendentry{BOC theory}
\end{semilogxaxis}
\end{tikzpicture}
\vspace{-3pt}
\caption{Finite-population scaling, $M=1024$ ($3000$ trials per point).}
\label{fig:scaling}
\end{figure}
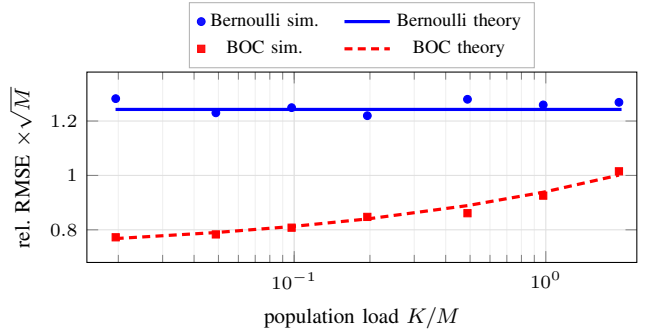

Fig.~\ref{fig:scaling} varies the population at fixed frame length, with the BOC load minimizing \eqref{eq:qlaw} with $\bar\sigma_r^2=1/6$ at each $K$. Simulation follows the law from $K/M=0.02$ to nearly $2$, so the improvement persists well outside the $K/M\to0$ regime. At $K/M\simeq0.1$ the normalized RMSE improves from about $1.25$ for Bernoulli activation to $0.81$ for BOC, and at $K/M\simeq2$ from $1.27$ to $1.02$. The margin narrows because the two terms of \eqref{eq:qlaw} scale differently in $K$: the collision term is set by the load alone, whereas the rounding term $K\bar\sigma_r^{2}/M$ grows with the population, so the advantage that quota encoding buys is progressively spent on rounding. Remark~\ref{rem:regime} bounds where this ends, and the practical consequence is that BOC is worth its extra encoder complexity mainly for $K\lesssim M$; beyond that the two encoders converge and Bernoulli activation is the simpler choice.

At $M=1024$ and $K=100$, Bernoulli activation performs $1024$ Bernoulli draws per device and transmits $16.6\pm3.8$ bursts at the simulated operating point. BOC instead forms one quota and uses about $6.1\pm0.5$ placement draws per device. The lower BOC estimation error is therefore obtained with less local random-number generation and less transmit activity, with the largest difference in the $K\lesssim M$ regime just described.

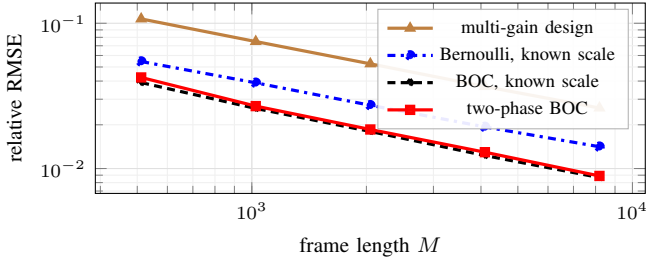
\begin{figure}[t]
\centering
\begin{tikzpicture}
\begin{loglogaxis}[twcaxis, xlabel={frame length $M$}, ylabel={relative RMSE},
 legend style={font=\scriptsize,fill opacity=0.92,text opacity=1,at={(0.97,0.97)},anchor=north east}]
\addplot[brown,mark=triangle*,mark size=1.45,very thick] coordinates {(512,0.107125)(1024,0.074836)(2048,0.052572)(4096,0.036687)(8192,0.026101)};
\addlegendentry{multi-gain design}
\addplot[blue,mark=*,mark size=1.4,very thick,dashdotted] coordinates {(512,0.054531)(1024,0.038903)(2048,0.027254)(4096,0.019340)(8192,0.014127)};
\addlegendentry{Bernoulli, known scale}
\addplot[black,mark=diamond*,mark size=1.35,very thick,densely dashed] coordinates {(512,0.039055)(1024,0.025923)(2048,0.017996)(4096,0.012294)(8192,0.008667)};
\addlegendentry{BOC, known scale}
\addplot[red,mark=square*,mark size=1.35,very thick] coordinates {(512,0.042357)(1024,0.026893)(2048,0.018570)(4096,0.012956)(8192,0.008891)};
\addlegendentry{two-phase BOC}
\end{loglogaxis}
\end{tikzpicture}
\vspace{-3pt}
\caption{Operation without scale knowledge, $D=10^3$, $K=200$ ($1800$ trials per point).}
\label{fig:blind}
\end{figure}

Fig.~\ref{fig:blind} compares two-phase BOC against three references: BOC and Bernoulli activation with the scale known, and the non-adaptive guarded geometric design of Section~\ref{sec:scale} with $Q=13$ and $\beta=2$, whose measured coefficient is $2.38$ against the $2.341$ predicted by \eqref{eq:ladderlaw}. The $6\%$ probe overhead of two-phase BOC is included in $M$. Acquiring the gain rather than assuming it costs little. At $M=2048$, two-phase BOC yields $1.86\%$ against $1.80\%$ for known-scale BOC, $2.73\%$ for known-scale Bernoulli activation and $5.26\%$ for the non-adaptive multi-gain design. The probe closes almost the whole gap for the reason given in Section~\ref{sec:blind}: the coarse estimate it returns is well within the tolerance of \eqref{eq:mismatch}, so $6\%$ of the REs buys it. The non-adaptive design has no such escape, since Theorem~\ref{thm:conserve} applies to any fixed gain distribution and its penalty is paid on every frame. Two-phase operation is therefore preferable whenever a single downlink broadcast is available.
\vspace{-10pt}
\subsection{Imperfect activity detection}

\begin{figure}[t]
\centering
\begin{tikzpicture}
\begin{axis}[twcaxis, xlabel={average miss probability $\bar\varepsilon$}, ylabel={relative bias of the mean},
 ymin=-0.36,ymax=0.035,
 legend style={font=\scriptsize,fill opacity=0.92,text opacity=1,at={(0.5,1.02)},anchor=south,legend columns=2},
 ytick={-0.3,-0.2,-0.1,0}]
\addplot[blue,mark=*,mark size=1.45,very thick] coordinates
{(0.02,-0.0192)(0.05,-0.0474)(0.1,-0.0995)(0.15,-0.1449)(0.2,-0.2157)(0.25,-0.2560)(0.3,-0.3019)(0.35,-0.3322)};
\addlegendentry{uncorrected}
\addplot[red,mark=square*,mark size=1.35,very thick] coordinates
{(0.02,0.0012)(0.05,0.0009)(0.1,0.0045)(0.15,0.0026)(0.2,-0.0061)(0.25,0.0108)(0.3,0.0005)(0.35,0.0003)};
\addlegendentry{self-normalized}
\addplot[black,densely dotted,thick] coordinates {(0.02,0)(0.35,0)};
\end{axis}
\end{tikzpicture}
\vspace{-3pt}
\caption{Relative bias under heterogeneous misses, $\varepsilon_k\sim\mathrm{Unif}[0,2\bar\varepsilon]$, $K=200$, $M=2048$, $M_r=512$, $n_0=2$.}
\label{fig:miss}
\end{figure}
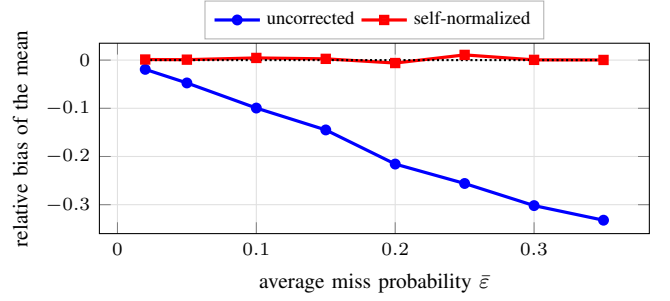

Fig.~\ref{fig:miss} shows the effect of heterogeneous misses on quota encoding, with and without the self-normalization of Section~\ref{ssec:selfnorm}. Uncorrected, the aggregate is attenuated almost in proportion to the average miss rate: the mean is $14.5\%$ low at $\bar\varepsilon=0.15$ and $30\%$ low at $\bar\varepsilon=0.30$. The attenuation is not a detector artifact; rather, it reflects a change in the effective target: by Lemma~\ref{lem:miss} the uncorrected estimator converges to $\sum_k(1-\varepsilon_k)x_k$ rather than to $S$. Self-normalization holds the residual bias within about $\pm1\%$ without using the individual SNRs, as Proposition~\ref{prop:selfnorm} predicts: the reference group estimates the same effective count $\sum_k(1-\varepsilon_k)$ that appears in the numerator, so it cancels in \eqref{eq:ratio}. The correction costs the factor $2.31$ that Proposition~\ref{prop:split} gives for the $M_r=M/4$ split used here, which is the price of splitting the frame rather than of the misses themselves.
\vspace{-10pt}
\subsection{Comparison with external aggregation baselines}

\begin{figure}[t]
\centering
\begin{tikzpicture}
\begin{loglogaxis}[twcaxis, xlabel={physical channel uses $R$}, ylabel={relative RMSE},
 xmin=50,xmax=5200,ymin=0.009,ymax=0.4,
 legend style={font=\scriptsize,fill opacity=0.93,text opacity=1,at={(0.5,1.02)},anchor=south,legend columns=3,
 /tikz/every even column/.append style={column sep=0.15cm}}]
\addplot[mark=o,blue,densely dashed,thick,mark size=1.2] coordinates {(64,0.1677)(256,0.0748)(1024,0.0390)(4096,0.0203)};
\addlegendentry{Bernoulli ODC}
\addplot[mark=square*,red,thick,mark size=1.0] coordinates {(64,0.125696)(256,0.055664)(1024,0.025458)(4096,0.012091)};
\addlegendentry{BOC}
\addplot[mark=triangle*,brown,thick,mark size=1.3] coordinates {(64,0.1696)(256,0.1334)(1024,0.1221)(4096,0.1248)};
\addlegendentry{affine NC-OAC}
\addplot[mark=diamond*,black,thick,mark size=1.2] coordinates {(64,0.175329)(256,0.089309)(1024,0.044179)(4096,0.022086)};
\addlegendentry{REED}
\addplot[mark=o,blue,dotted,very thick,mark size=1.2] coordinates {(64,0.17160)(256,0.11389)(1024,0.05525)(4096,0.02698)};
\addlegendentry{ODC, 20\,dB}
\addplot[mark=square,red,dotted,very thick,mark size=1.2] coordinates {(64,0.14303)(256,0.08242)(1024,0.03718)(4096,0.01780)};
\addlegendentry{BOC, 20\,dB}
\addplot[brown,densely dotted,very thick] coordinates {(64,0.1165)(4096,0.1165)};
\addlegendentry{fading floor}
\end{loglogaxis}
\end{tikzpicture}
\vspace{-3pt}
\caption{External comparison at matched physical channel uses, $K=100$. Curves without a 20-dB label use ideal activity decisions; the 20-dB ODC/BOC curves use the Gamma detector.}
\label{fig:baseline}
\end{figure}
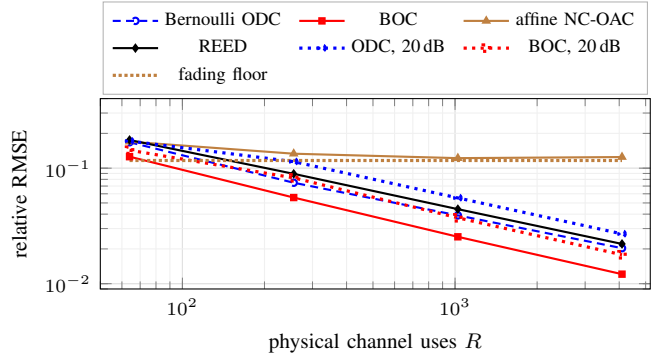

All four schemes in Fig.~\ref{fig:baseline} estimate the same arithmetic sum, and are compared at equal physical channel uses though not at equal radiated energy. With ideal activity decisions, BOC has the lowest error over the common resource range. At $1024$ physical channel uses the relative RMSEs are $12.2\%$ for affine NC-OAC, $4.42\%$ for REED, $3.90\%$ for Bernoulli ODC and $2.55\%$ for BOC; at $4096$, REED reaches $2.21\%$ and BOC $1.21\%$, while affine NC-OAC barely moves. These curves isolate the aggregation rule: with detection idealized, the only difference between the schemes is what the receiver measures, and the occupancy observable carries no channel weight while the energy observables do. The comparison is also not symmetric in what the transmitters know: both energy baselines are given exact average channel powers here.

The $20$~dB curves add the detector of Section~\ref{sec:det}, with $N_s\in\{1,2,4\}$ selected at each budget. Detection costs the two occupancy encoders differently. BOC stays below every baseline, attaining $14.3\%$, $8.24\%$, $3.72\%$ and $1.78\%$ at $64$, $256$, $1024$ and $4096$ physical channel uses against $17.53\%$, $8.93\%$, $4.42\%$ and $2.21\%$ for REED. The margin is a factor of $1.08$ to $1.24$ across this range, against $1.4$ to $1.8$ for the same curves under ideal decisions. The choice of $N_s$ matters, by an amount Fig.~\ref{fig:snr} quantifies. Bernoulli ODC alone also remains above REED once detection is included, reaching $11.4\%$, $5.53\%$ and $2.70\%$ at $256$, $1024$ and $4096$ physical channel uses; at this operating point only quota encoding yields an end-to-end margin over REED.

\begin{figure}[t]
\centering
\begin{tikzpicture}
\begin{semilogyaxis}[twcaxis, xlabel={mean receive SNR (dB)}, ylabel={relative RMSE},
 xmin=15, xmax=37, ymin=0.016, ymax=0.22, legend columns=2,
 legend style={at={(0.97,0.97)},anchor=north east}]
\addplot[mark=square*,red,thick,mark size=1.0] coordinates {(16,0.19282)(20,0.08803)(24,0.04030)(28,0.02279)(32,0.01843)(36,0.01775)};
\addlegendentry{BOC, $N_s=1$}
\addplot[mark=square,red,densely dashed,thick,mark size=1.3] coordinates {(16,0.04878)(20,0.02631)(24,0.02562)(28,0.02559)(32,0.02484)(36,0.02524)};
\addlegendentry{BOC, $N_s=2$}
\addplot[mark=square,red,dotted,very thick,mark size=1.0] coordinates {(16,0.03741)(20,0.03626)(24,0.03807)(28,0.03806)(32,0.03695)(36,0.03792)};
\addlegendentry{BOC, $N_s=4$}
\addplot[mark=o,blue,densely dashed,thick,mark size=1.2] coordinates {(16,0.05661)(20,0.04010)(24,0.03872)(28,0.03316)(32,0.02870)(36,0.02721)};
\addlegendentry{Bernoulli ODC}
\addplot[mark=diamond*,black,thick,mark size=1.2] coordinates {(16,0.03140)(20,0.03104)(24,0.03131)(28,0.03122)(32,0.03129)(36,0.03144)};
\addlegendentry{REED}
\addplot[mark=triangle*,brown,thick,mark size=1.3] coordinates {(16,0.11714)(20,0.11845)(24,0.11755)(28,0.11858)(32,0.11772)(36,0.11697)};
\addlegendentry{affine NC-OAC}
\addplot[gray,densely dotted,very thick,forget plot] coordinates {(15,0.0175)(37,0.0175)};
\end{semilogyaxis}
\end{tikzpicture}
\caption{Relative RMSE versus mean receive SNR, $R=2048$ physical channel uses, $K=100$. Horizontal dotted line: ideal-detection floor $1.75\%$.}
\label{fig:snr}
\end{figure}
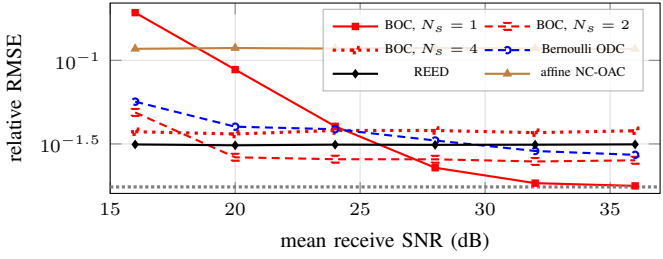

In Fig.~\ref{fig:snr}, BOC uses $M=R/N_s$ logical REs as in Section~\ref{sec:det}, Bernoulli ODC is shown at its best $N_s$, and REED uses the paired construction of~\cite{ChenREED26}, repeated over $R/2$ independent pairs. The two families are limited by different quantities, which is why their SNR dependence differs in kind. Neither energy baseline has receiver noise as its dominant error term: affine NC-OAC carries the frozen channel gains of Remark~\ref{rem:shot}, and REED replaces that floor by the fading self-noise of~\cite[Prop.~2]{ChenREED26}, which is reduced by averaging over $R/2$ pairs to the high-SNR limit $\sqrt{2/R}=3.1\%$. Neither term contains $\sigma_z^{2}$, so raising the transmit power moves neither curve. In the occupancy domain fading enters only through the miss probability $\varepsilon_k$ of \eqref{eq:thin}, which vanishes as the SNR grows. The error law then reduces to \eqref{eq:qlaw}, which depends on the load alone. The occupancy curves therefore descend until they meet the resolution set by the number of logical REs.

That difference dictates how the diversity order should be chosen. Each BOC curve saturates at its own value of \eqref{eq:qlaw} evaluated at $M=R/N_s$, so diversity buys detection reliability with logical REs at a fixed budget, and the useful setting is the smallest $N_s$ whose miss floor lies below that resolution floor. The crossover moves with SNR: $N_s=4$ is preferable at $16$~dB, $N_s=2$ at $20$ and $24$~dB, and $N_s=1$ from $28$~dB, where the $8.80\%$ of $N_s=1$ at $20$~dB has fallen to $2.28\%$. Beyond $32$~dB the $N_s=1$ curve flattens at the dotted line, the ideal-detection value $1.75\%$ of \eqref{eq:qlaw} at $M=R$: misses have become negligible and the collision and rounding terms alone remain, so further transmit power buys nothing and the limit is resolution rather than detection. The practical rule is to set $N_s$ from the link budget rather than fixing it once, since choosing it too small costs an order of magnitude near the detection limit while choosing it too large costs a factor $\sqrt{N_s}$ at high SNR.

Across the schemes, BOC improves from $3.74\%$ at $16$~dB to $1.78\%$ at $36$~dB and Bernoulli ODC from $5.66\%$ to $2.72\%$, while both baselines stay at their floors. The margins therefore grow with SNR: BOC is $4.5\times$ more accurate than affine NC-OAC at $20$~dB and $6.6\times$ at $36$~dB, and $1.2$ to $1.8\times$ more accurate than REED from $20$ to $36$~dB. The plotted range begins at $16$~dB because the occupancy schemes require each burst to be detected individually, whereas the energy baselines measure an aggregate and therefore benefit from the combined power of all $K$ devices; below roughly $16$~dB missed activations dominate and the energy schemes are preferable, which is the operating regime of~\cite{ChenREED26}. The two margins differ in kind. Affine NC-OAC has a fading floor no resource budget closes, so its gap to BOC keeps growing. The gap to REED saturates instead, because both schemes scale as $R^{-1/2}$: comparing that $3.1\%$ limit with the $1.75\%$ reached here bounds the ratio at $1.8$, consistent with the factor $2$ that Corollary~\ref{cor:load} gives under ideal detection with $K/M\to0$. What does not saturate is the requirement: REED scales each device's transmit amplitude by the reciprocal of its average channel amplitude $\mu_k$, where $\mu_k^{2}=\E|h_k|^{2}$, and by the square root of its value, so it needs both a calibrated $\mu_k$ and linear amplitude control, and spends two REs per scalar with one carrying only noise for nonnegative data. BOC uses no channel quantity, transmits a fixed-amplitude burst or nothing, and needs a single RE group rather than the paired positive/negative groups.

\begin{figure}[t]
\centering
\begin{tikzpicture}
\begin{semilogyaxis}[twcaxis, xlabel={average-channel-power calibration error (dB)}, ylabel={relative RMSE},
 xmin=-0.1, xmax=3.1, ymin=0.02, ymax=0.30,
 legend style={at={(0.03,0.95)},anchor=north west}]
\addplot[mark=triangle*,brown,thick,mark size=1.2] coordinates {(0,0.11770)(0.5,0.11925)(1,0.12379)(1.5,0.13135)(2,0.14204)(3,0.17396)};
\addlegendentry{affine NC-OAC}
\addplot[mark=diamond*,black,thick,mark size=1.2] coordinates {(0,0.03114)(0.5,0.03405)(1,0.04138)(1.5,0.05182)(2,0.06452)(3,0.09605)};
\addlegendentry{REED}
\addplot[mark=o,blue,densely dashed,thick,mark size=1.2] coordinates {(0,0.03819)(0.5,0.03961)(1,0.03858)(1.5,0.03977)(2,0.03896)(3,0.04018)};
\addlegendentry{Bernoulli ODC}
\addplot[mark=square*,red,thick,mark size=1.0] coordinates {(0,0.02740)(0.5,0.02674)(1,0.02684)(1.5,0.02752)(2,0.02755)(3,0.02644)};
\addlegendentry{BOC}
\end{semilogyaxis}
\end{tikzpicture}
\caption{Sensitivity to average-channel-power calibration error, $R=2048$, $K=100$, $20$~dB mean receive SNR.}
\label{fig:calib}
\end{figure}
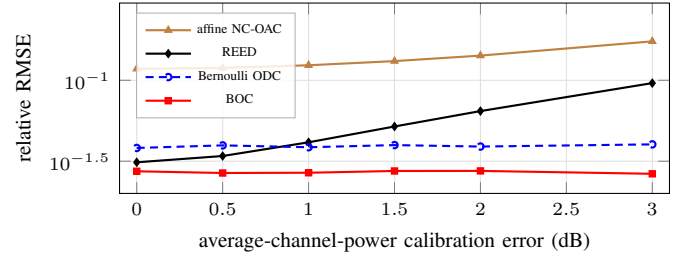

Fig.~\ref{fig:calib} withdraws the exact $\mu_k$ that Figs.~\ref{fig:baseline} and~\ref{fig:snr} granted the energy baselines: each device using them applies $\hat\mu_k$ in place of $\mu_k$, with a log-normal error of standard deviation $\sigma_{\rm dB}$ normalized so that the transmit scaling is unbiased. REED is exact when $\mu_k$ is known exactly, and the figure starts from that point. A calibration error multiplies device $k$'s contribution by $(\mu_k/\hat\mu_k)^2$, a weight that is common to all of its pairs, so pair diversity cannot average it away and it acts as a second error floor alongside the fading-induced floor that REED is designed to suppress. The added relative standard deviation is $\sqrt{e^{s^{2}}-1}\,\lVert x\rVert_2/S$ with $s=(\ln 10/10)\sigma_{\rm dB}$ and $\lVert x\rVert_2=(\sum_k x_k^2)^{1/2}$, which evaluates to $2.7\%$, $5.6\%$ and $9.0\%$ at $1$, $2$ and $3$~dB for the population used here and matches the simulated curve. The corresponding REED error rises from $3.11\%$ to $4.14\%$, $6.45\%$ and $9.60\%$, and affine NC-OAC from $11.8\%$ to $17.4\%$, while BOC stays within $2.6$--$2.8\%$ and Bernoulli ODC within $3.8$--$4.1\%$ across the range, because no channel quantity enters \eqref{eq:identity} or \eqref{eq:qidentity}. The added standard deviation scales as $\lVert x\rVert_2/S=\Theta(K^{-1/2})$, so it matters most for small populations.
\vspace{-5pt}
\section{Conclusions}
\label{sec:concl}
In this paper, we proposed occupancy-domain computation, which changes the observable of over-the-air computation from an amplitude or an energy to a new dimension: a binary busy/idle decision on each resource element, where fading enters only through whether an activation is detected, and no device needs instantaneous or statistical channel knowledge.

The analysis identified what governs accuracy once the observable is binary. Exponential Bernoulli activation follows directly from the occupancy identity and is asymptotically efficient, and we established a conservation law that fixes the cost of operating over an unknown scale. Balanced occupancy computation then improves on it by fixing the placement-draw quota from the data. This removes a source of noise that carries no information about the aggregate, and it places the design outside that law. We established its exact identity, its error law and operating point, its leading-order asymptotic dominance over Bernoulli activation at every load, and a deviation bound valid at finite frame length and finite population.

Against the non-coherent energy baselines, the proposed approach is more accurate while requiring no channel knowledge at all, and it retains that accuracy as their calibration degrades. The two occupancy encoders themselves trade off differently: quota encoding is more accurate and requires less transmit activity, while Bernoulli activation reveals less about an individual value and becomes the simpler choice once the population approaches the frame length. Structured rather than uniform placement, superimposing several coordinates on one frame, the multi-cell case, and formal privacy guarantees for an observable that is anonymous by construction remain open.

\appendices

\section{Proof of Theorem~\ref{thm:bern}}
\label{app:bern}
The variables $B_m$ are i.i.d.\ $\mathrm{Bern}(1-q)$, so the likelihood depends on the frame only through $N_0$. The unconstrained binomial MLE is $\hat q=N_0/M$. Constraining $S\ge0$ is equivalent to $q\le1-\alpha$, which gives the outer clipping in \eqref{eq:bestimator}; the $N_0=0$ truncation supplies the finite upper cap used in the main text. Both clipping events have exponentially small probability at fixed $q\in(0,1)$, so they do not change the stated asymptotic second moment. The delta method with $|dS/dq|=1/(aq)$ gives $\Var(\Shat)=(1-q)/(Ma^{2}q)$, and substituting $a=\lambda/S$ yields \eqref{eq:bclt}. One Bernoulli observation carries Fisher information $(dq/dS)^{2}/[q(1-q)]=a^{2}q/(1-q)$ about $S$, whose frame total inverts to that variance, so the bound is attained.

Setting $\alpha=0$, so that $q=e^{-\lambda}$, gives relative variance $r(\lambda)/M$ with $r(\lambda)=(e^{\lambda}-1)/\lambda^{2}$, which is \eqref{eq:brmse}. Its derivative vanishes if and only if $\lambda e^{\lambda}=2(e^{\lambda}-1)$. Since $\phi(\lambda)=\lambda e^{\lambda}-2e^{\lambda}+2$ has $\phi(0)=0$ and $\phi'(\lambda)=(\lambda-1)e^{\lambda}$, it decreases on $(0,1)$ and increases afterwards, so the positive root is unique, $\lambda^\star=1.59362$, with $\sqrt{r(\lambda^\star)}=1.24263$. At $\lambda=0.8$ and $\lambda=3.2$ we obtain $\sqrt r=1.3838$ and $1.5159$, corresponding to increases of $11.4\%$ and $22.0\%$. \hfill$\blacksquare$

\section{Proof of Theorem~\ref{thm:conserve}}
\label{app:conserve}
(i) With $t=\ln a$ and $s=\ln S$, the per-RE information about $\ln S$ at gain $a$ is $J(aS)$, so $I(s)=\int J(e^{t+s})f(dt)$, which is a convolution. By Tonelli's theorem, $\int_\R I(s)\,ds=\int_\R J(e^{v})\,dv=\int_0^\infty J(u)\,du/u$, which equals $\Gamma(2)\zeta(2)$ by $\int_0^\infty u^{s-1}/(e^{u}-1)\,du=\Gamma(s)\zeta(s)$ at $s=2$. The value is independent of $f$ because $f$ is a probability measure.

(ii) The minimum of $I$ over an interval of length $\ln D$ is at most its average over that interval, which by \eqref{eq:conserve} is at most $\pi^{2}/(6\ln D)$. The Cram\'er--Rao bound over $M$ conditionally independent REs gives the RMSE claim.

(iii) The restriction to an interior aggregate in the statement lets the finite sum over the $Q$ groups be replaced by the two-sided profile $G_\beta(u)=\sum_{q\in\mathbb Z}J(\beta^{q}u)$, the omitted tails being the edge effects the statement excludes. Writing $u=e^{t}$ makes $G_\beta(e^{t})$ periodic with period $\ln\beta$. Unfolding the periodization, its $n$th Fourier coefficient is $(\ln\beta)^{-1}\int_0^\infty v^{1-i\omega_n}/(e^{v}-1)\,dv=(\ln\beta)^{-1}\Gamma(2-i\omega_n)\zeta(2-i\omega_n)$ with $\omega_n=2\pi n/\ln\beta$, so that
\[
G_\beta(u)=\frac{1}{\ln\beta}\sum_{n\in\mathbb Z}\Gamma(2-i\omega_n)\zeta(2-i\omega_n)\,u^{\,i\omega_n}.
\]
Its zero-frequency term is $\pi^{2}/(6\ln\beta)$, and the remaining harmonics are bounded by $\delta_\beta\le\frac{12}{\pi^{2}}\sum_{n\ge1}|\Gamma(2+i\omega_n)\zeta(2+i\omega_n)|$, which gives $\delta_2=6.5\times10^{-5}$ and $\delta_4\simeq2.0\times10^{-2}$. Substituting into the frame information about $\ln S$, which for equal allocation is $\sum_qM_qJ(a_qS)=(M/Q)G_\beta(a_0S)$, and invoking standard asymptotic efficiency of the regular MLE for the independent gain groups, gives \eqref{eq:ladderlaw}. \hfill$\blacksquare$

\section{Proof of Theorem~\ref{thm:qlaw}}
\label{app:qlaw}
Condition on $N$. Extending Lemma~\ref{lem:qidentity} to pairs, $\Prb\{B_m=B_{m'}=0\mid N\}=(1-\alpha)^{2}(1-2/M)^{N}$ for $m\ne m'$, since a burst misses two given REs with probability $1-2/M$. With $\alpha\to0$, $q=(1-1/M)^{N}$ and $q_2=(1-2/M)^{N}$, we have $\E[N_0]=Mq$ and $\Var(N_0)=Mq+M(M-1)q_2-M^{2}q^{2}$. Expanding both, we obtain $q=e^{-\lambda}(1+O(M^{-1}))$ and $q_2-q^{2}=-\lambda e^{-2\lambda}/M+O(M^{-2})$, so $\Var(N_0)=Me^{-\lambda}[1-(1+\lambda)e^{-\lambda}]+O(1)$. The map $N\mapsto q$ has $|dN/d\log q|=1/|\log(1-1/M)|=M+O(1)$, so the delta method gives $\Var(\widehat N\mid N)=\Var(N_0)/q^{2}=M(e^{\lambda}-1-\lambda)+O(1)$. This is valid on the event $N_0>0$, which has probability tending to one: the occupancy indicators are dependent, but uniform placement makes them negatively associated, so $\Prb\{N_0=0\}\le\prod_m\Prb\{B_m=1\}=(1-q)^{M}\to0$ under $Mq\to\infty$. Rounding is independent of the placement, so by the law of total variance $\Var(\widehat N)=M(e^{\lambda}-1-\lambda)+\sum_k\sigma_{r,k}^{2}+O(1)$, where expanding the conditional variance about $\lambda=\E[N]/M$ contributes $O(\Var(N)/M)=O(1)$ under the hypothesis $K/M=O(1)$. Dividing by $\E[N]^{2}=M^{2}\lambda^{2}$ gives \eqref{eq:qlaw}. Repeating the computation with $q=(1-\alpha)e^{-\lambda}$ replaces $e^{\lambda}-1-\lambda$ by $e^{\lambda}/(1-\alpha)-1-\lambda$. \hfill$\blacksquare$

\section{Proof of Theorem~\ref{thm:finite}}
\label{app:finite}
Write $\widehat N=a\Shat$, so that $\widehat N=\log(N_0/M)/\log(1-1/M)$ when $N_0>0$ and $\alpha=0$. Conditioned on $N$, the idle count $N_0$ is a function of the independent placements $P_1,\dots,P_N$, each uniform on $\{1,\dots,M\}$. Changing one placement can free one RE and occupy another, and these act in opposite directions, so the net change in $N_0$ lies in $\{-1,0,1\}$ and the bounded-difference constants are $c_i=1$. Since $\E[N_0\mid N]=Mq$ by Lemma~\ref{lem:qidentity}, McDiarmid's inequality gives $\Prb\{|N_0-Mq|\ge t\mid N\}\le2e^{-2t^{2}/N}$.

On the event $|N_0-Mq|\le t_\delta$ with $t_\delta<Mq$ we have $N_0>0$, and since $N=\log q/\log(1-1/M)$, the difference is $\widehat N-N=\log(1+z)/\log(1-1/M)$ with $z=(N_0-Mq)/(Mq)$. Using $|\log(1+z)|\le|z|/(1-|z|)$ and $1/|\log(1-1/M)|\le M+1$ for $M\ge2$ gives $|\widehat N-N|\le(M+1)\frac{t_\delta/(Mq)}{1-t_\delta/(Mq)}$. Separately, $N-aS=\sum_k(n_k-ax_k)$ is a sum of $K$ independent, zero-mean terms, each taking one of the two values $-\phi_k$ and $1-\phi_k$ and therefore of range one, so Hoeffding's inequality gives $\Prb\{|N-aS|\ge s_\delta\}\le2e^{-2s_\delta^{2}/K}$. Since $\Shat=\widehat N/a$ and $aS=M\lambda$, the triangle inequality gives $|\Shat-S|/S\le(|\widehat N-N|+|N-aS|)/(M\lambda)$, and a union bound at level $\delta/2$ for each event gives \eqref{eq:finitebound}. Substituting $q=e^{-\lambda}(1+O(M^{-1}))$ and $N=M\lambda$ yields the asymptotic form. \hfill$\blacksquare$
\vspace{-5pt}
\section{Proof of Lemma~\ref{lem:miss}}
\label{app:miss}
For Bernoulli activation, a device contributes no effective activation with probability $1-(1-\varepsilon_k)p_k=\varepsilon_k+(1-\varepsilon_k)e^{-ax_k}$, which gives \eqref{eq:bernmiss}. Applying the original logarithmic decoder yields the first expression in \eqref{eq:bernmisslimit}. Expanding $e^{-ax_k}=1-ax_k+O(a^2x_k^2)$ and then the logarithm gives the second expression.

For quota encoding, let $\tilde n_k\mid n_k\sim\mathrm{Bin}(n_k,1-\varepsilon_k)$, with thinning applied to the placement draws. A device transmits one physical burst per distinct selected RE, so draw thinning and burst thinning differ on repeated selections. The expected number of repeated draws is $n_k^{2}/(2M)+O(n_k^{3}/M^{2})$, a fraction $n_k/(2M)$ of the quota, so their relative difference is $O(\max_kn_k/M)$; we assume $\max_kn_k=o(M)$, which holds whenever no single device carries a non-vanishing fraction of the aggregate. The laws of total expectation and total variance give
\[
\begin{aligned}
\E[\tilde n_k]&=(1-\varepsilon_k)ax_k,\\
\Var(\tilde n_k)&=ax_k\varepsilon_k(1-\varepsilon_k)
 +(1-\varepsilon_k)^2\sigma_{r,k}^2.
\end{aligned}
\]
Hence $\E[\tilde N]=aS(1-\varepsilon_x)=M\lambda_e$ and
\[
\Var(\tilde N)=M\lambda(\varepsilon_x-\varepsilon_{2,x})+\Sigma_r.
\]
A uniformly placed quota thinned by independent erasures is again uniformly placed, so Lemma~\ref{lem:qidentity} applies conditional on $\tilde N$. Repeating Appendix~\ref{app:qlaw} at effective load $\lambda_e$ and adding the variance of $\tilde N$ gives \eqref{eq:thin}. \hfill$\blacksquare$

\end{document}